\documentclass[a4paper]{amsart}
\usepackage{amsmath,amssymb,amsthm}
\usepackage{float}
\usepackage{tikz}
\usepackage[hidelinks]{hyperref}
\usepackage{enumerate}
\usepackage{graphicx}
\usetikzlibrary{calc}
\newtheorem{definition}{Definition}
\newtheorem{theorem}{Theorem}
\newtheorem{lemma}{Lemma}
\newtheorem{proposition}{Proposition}
\newtheorem{remark}{Remark}
\newtheorem{note}{Note}
\newtheorem{example}{Example (Non-empty blind spot)}
\newtheorem{procedure}{Procedure}
\newtheorem{assumption}{Assumption}
\newtheorem{corollary}{Corollary}
\usepackage{fontspec}
\usepackage{unicode-math}
\usepackage[a4paper, margin = 31.4mm]{geometry}

\begin{document}
  
\title{Detection coherence of tests}
\author{M. Grend\'ar}
\address{1~Laboratory of Bioinformatics and Biostatistics, Biomedical Centre Martin,
Jessenius Faculty of Medicine, Comenius University in Bratislava, Slovakia.
2~Laboratory of Theoretical Methods, Institute of Measurement Science, Slovak Academy of Sciences, Bratislava, Slovakia.
3~Bioptic Laboratory Ltd., Plzen, Czech Republic.}

\begin{abstract}
There exist tests calibrated under a null narrower than the one implied by
their test statistic -- the detection-null set. The part of the detection-null set
not in the null is the test's blind spot. A framework for assessing detection
coherence -- whether a test's blind spot is empty -- is introduced, built on new
concepts of calibration statistic, detector, detection functional, detection-null set,
and blind spot. It is demonstrated that the Wilcoxon--Mann--Whitney, Kruskal--Wallis,
Friedman, and logrank tests are detection incoherent as unrestricted tests. Detection
incoherent tests may become coherent under restrictions that make the blind spot empty,
though domain restriction is not a reliable route to coherence in practice,
since the boundary of the coherence-restoring domain is unknown and unverifiable.
The Kolmogorov--Smirnov and Zaremba tests, and the Maximum Mean Discrepancy test with a
characteristic kernel, are shown to be universally detection coherent. Due to the blind spot,
a detection incoherent test misses discoveries when non-rejecting and yields spurious
discoveries when rejecting, in both cases regardless of sample size. Detection incoherent
tests should be abandoned.
\end{abstract}

\maketitle

%
\section{Introduction}

There exist tests calibrated under a null
narrower than the one their test statistic implies.
The Wilcoxon--Mann--Whitney (WMW) test
provides a simple illustration of such a detection incoherent test:
it is calibrated under $\mathcal{H}_0 = \{(F_1, F_2)\mathpunct{:} F_1 = F_2\}$,
but its detector, the empirical AUC, $\mathrm{eAUC} = U/(n_1 n_2)$,
implies calibration under the broader detection-null set
$\mathcal{D}_0 = \{(F_1, F_2)\mathpunct{:} \mathrm{AUC} = 1/2\}$.
Since the test's blind spot $\mathcal{S}_{\mathcal{P}}$ over the distribution space $\mathcal{P}$, 
$\mathcal{S}_{\mathcal{P}} = \mathcal{D}_0
\setminus \mathcal{H}_0 = \{(F_1, F_2)\mathpunct{:} \mathrm{AUC} = 1/2\}\setminus \{(F_1, F_2)\mathpunct{:} F_1 = F_2\}$ is non-empty,
the test is detection incoherent.
Asymptotically, a non-significant result from the WMW
test can therefore come either from the valid null or
from the blind spot -- a flaw of the WMW test recognized
by Gnedenko~\cite{gnedenko1958} in~1958.

This paper introduces a framework, \emph{detection coherence}, for assessing
whether a test has a blind spot. The framework is built on new concepts: calibration
statistic, detector, detection functional, detection-null set,
and blind spot.
The detector $\hat D_n$ is the component of the calibration statistic $T_n$
whose convergence in probability to the detection functional,
$\hat{D}_n \xrightarrow{p} \theta(P)$ for $P \in \mathcal{P}$,
drives $T_n$'s convergence in distribution under $\mathcal{H}_0$, $T_n \xrightarrow{d} \mathcal{L}_0$.
The detection-null set is $\mathcal{D}_0 = \{P\mathpunct{:} \theta(P) = \theta^*\}$,
where $\theta^* = \theta(P)$ for any $P \in \mathcal{H}_0$ is the null value
of the detection functional.
A test is detection incoherent when the detection-null set is larger than the null,
$\mathcal{H}_0\subsetneq \mathcal{D}_0$. Its blind spot,
$\mathcal{S}_{\mathcal P} = \mathcal{D}_0\setminus\mathcal{H}_0$, the set of
distributions belonging to the detection-null but not to the null, is non-empty.

It is demonstrated that the four major nonparametric tests
(the Wilcoxon--Mann--Whitney (WMW), Kruskal--Wallis (KW), Friedman, and
logrank tests) are detection incoherent as unrestricted tests.

A detection incoherent test may become detection
coherent when restricted to a domain that excludes its blind spot. Domain restrictions
under which each test achieves coherence are identified:
monotone likelihood ratio for WMW, ordered hazards for logrank,
location-shift for KW and block-consistent stochastic ordering for Friedman.
These are sufficient conditions, and not maximal ones. The largest domain on which a test is coherent 
is $\mathcal{R}^*$, the set on which $\theta(P) = \theta^*$ forces $P \in \mathcal{H}_0$; 
each named model sits strictly inside it, and outside the model but still inside $\mathcal{R}^*$ the test remains coherent. 
How far inside is not known for any of the four tests: the distance from the named model to the boundary of $\mathcal{R}^*$ is unknown, 
and membership cannot be verified from data. That is why these restrictions are sufficient conditions for coherence 
rather than reliable remedies for incoherence in practice -- a distinction developed in Section~\ref{Sect:discussion}.

Zaremba's~\cite{zaremba1962} test that calibrates the WMW statistic
under $\mathcal{H}_0 = \{(F_1, F_2)\mathpunct{:} \mathrm{AUC}=1/2\}$
instead of WMW's $\mathcal{H}_0 = \{(F_1, F_2)\mathpunct{:} F_1=F_2\}$,
the Kolmogorov--Smirnov~(KS) test, and the Maximum Mean Discrepancy (MMD) test with
a characteristic kernel are shown to be universally detection coherent.

Due to the blind spot, a detection incoherent test misses
discoveries when non-rejecting (a~false negative from the blind spot)
and yields spurious discoveries when rejecting (a true positive from the blind spot),
in both cases regardless of sample size; see~Sect.~\ref{Subsect:consequences}. 
Detection incoherent tests should be abandoned; see~Sect.~\ref{Sect:discussion}.

The paper is organized as follows. New concepts
(domain, unrestricted and restricted test, calibration
statistic, detector, detection functional,
detection-null value, detection-null
set, detection set, blind spot, and auxiliary argument) developed
to formulate detection coherence with respect to a domain and
universal detection coherence, are defined in
Section~\ref{Sect:definitions},
together with a result identifying the blind spot,
for two-sided alternatives, with the set of
distributions against which power does not grow
with sample size. A procedure for
systematic assessment of detection coherence, 
is presented in Section~\ref{Sect:recipe}. Detection
incoherence of the unrestricted WMW,
KW, Friedman, and logrank tests is
demonstrated in Section~\ref{Sect:major}.
For the KW and logrank tests, the detection functional depends on the allocation
of samples across groups, reflecting Brunner et al.'s~\cite{brunner2021} discovery
that the consistency region of certain rank-based tests in unbalanced designs
can depend on this allocation (Note~\ref{Note:auxiliary});
incoherence is demonstrated at every allocation.
Sufficient domain restrictions under which each of the four tests
achieves detection coherence are established in
Section~\ref{Sect:restriction}. Universal detection
coherence of the KS, Zaremba and MMD (with a characteristic kernel) tests is proved in
Section~\ref{Sect:two}.
Section~\ref{Sect:discussion}
situates the framework relative to classical nonparametric
theory, develops the consequences of detection incoherence
for inferences -- missed and spurious discoveries --
then examines the unreliability of domain restriction as a remedy in
practice, and discusses the framework's scope.
Section~\ref{sec:history} places the
framework in historical context, tracing prior work on the WMW
blind spot and its coherent reformulation.
Verification that the
detectors extracted for the four tests satisfy the
consistency and minimality conditions of
Definition~\ref{Def:detection}, including their
auxiliary-argument dependence, is provided in
Appendix~\ref{App:detector}.

%
\section{Definitions and Basic Properties}\label{Sect:definitions}

Let $X_1, \ldots, X_n$ be observations generated from a distribution
$P \in \mathcal P$, where $\mathcal P$ is the distribution space.

\begin{definition}[Domain]
A subset $\mathcal R \subseteq \mathcal P$ is called a domain.
\end{definition}

\begin{definition}[Test, null hypothesis, alternative, null calibration]\label{Def:nullhyp}
A statistical test defined over a domain $\mathcal R$ is characterized
by a null set $\mathcal H_0 \subseteq \mathcal R$, which partitions
the domain as $\mathcal R = \mathcal H_0 \cup \mathcal H_0^c$, where
$\mathcal H_0^c = \mathcal R \setminus \mathcal H_0$ is the alternative.

The test employs a test statistic $T_n^{\mathrm{test}}$ to form 
a calibration statistic $T_n = T_n(X_1,\ldots,X_n)$
whose reference distribution~$\mathcal{L}_0$ is derived under distributions
$P \in \mathcal H_0$ in order to construct a rejection rule $\varphi_n$,
where $\varphi_n = 1$ indicates rejection of $\mathcal H_0$ and
$\varphi_n = 0$ indicates non-rejection. This is referred to as
null calibration.
The calibration statistic is typically obtained from the
test statistic by centering and scale normalization, as for the WMW test,
where $T_n^{\mathrm{test}} = U$ and
$T_n = \sqrt{n}(\mathrm{eAUC} - 1/2)/\hat\sigma$; the two may also
coincide, as for the KW and Friedman tests.

We refer to $\mathcal H_0$ as the null hypothesis over domain
$\mathcal R$, and denote it also $H_0$ where appropriate.
\end{definition}

\begin{definition}[Unrestricted and domain-restricted tests]
A test defined over the full distribution space $\mathcal P$ is called
an unrestricted test. A test defined over a proper subset
$\mathcal R \subsetneq \mathcal P$ is called a domain-restricted test.
\end{definition}

\begin{definition}[Detector and detection functional]
\label{Def:detection}
Let $T_n$ be the calibration statistic of a test defined over
domain $\mathcal{R} \subseteq \mathcal{P}$. Suppose the calibration
statistic $T_n$ satisfies
$$
T_n \xrightarrow{d} \mathcal{L}_0 \quad \text{for all } P \in \mathcal{H}_0,
$$
and admits the representation
\begin{equation}\label{eq:representation}
T_n = \Phi(\hat{D}_n, \hat{\eta}_n)
\end{equation}
where
$\hat D_n$ is a scalar or vector-valued statistic,
$\hat\eta_n$ is a vector of nuisance estimators, and
$\Phi$ is a known functional.

The {detector} $\hat D_n$ is the component of~\eqref{eq:representation},
satisfying:
\begin{enumerate}[(i)]
\item \textit{Consistency:} $\hat{D}_n \xrightarrow{p} \theta(P)$
for some functional $\theta(P)$, for all $P\in\mathcal R$;
\item \textit{Minimality:} no non-injective function $f\mathpunct{:}\Theta\to\mathbb{R}$
(or $\mathbb{R}^d$, for vector-valued $\theta$) satisfies $f(\hat{D}_n)\xrightarrow{p}\theta(P)$
for all $P\in\mathcal R$,
where $\Theta = \{\theta(P)\mathpunct{:} P\in\mathcal R\}$ is the attainable range
of the detection functional.
\end{enumerate}
The functional $\theta(P)$ is called the {detection
functional}. The calibration statistic $T_n = \Phi(\hat{D}_n, \hat{\eta}_n)$
is the scalar functional of the detector
used for decision making.
\end{definition}

%
\begin{note}\label{Note:auxiliary}
Brunner et al.~\cite{brunner2021} discovered that the consistency
region of certain rank-based tests in unbalanced designs can
depend on the allocation of sample sizes across groups.
In the detection coherence framework this translates into the possibility
that the detection functional depends on an argument besides $P$;
i.e., $\theta(P; \xi)$ for an auxiliary argument $\xi$.
For the generalization of the framework to an auxiliary
argument, see Definitions~\ref{Def:auxiliary}, \ref{Def:auxiliary-generalization}.
\end{note}

%
The reader primarily
interested in the coherence framework may proceed
directly to Definition~\ref{Def:nullvalue}.

%
\begin{remark}[Digression: extraction of the detector]
For calibration statistics admitting a scaling-normalization-nuisance structure,
the extraction procedure is a backward engineering operation:
given $T_n \xrightarrow{d} \mathcal{L}_0$
under $H_0$, it recovers the underlying detector
$\hat{D}_n$ whose convergence in probability to the
detection functional $\theta(P)$, for $P\in\mathcal R$, drives that convergence
in distribution. For ratio statistics this backward
engineering is mechanical; for quadratic form statistics
it requires identifying the minimal vector pre-image,
guided by the consistency check. 
This is the extraction route of Note~\ref{Note:teststat}, where it is
contrasted with the normalization route.
\end{remark}

\begin{definition}[Extraction procedure]\label{Def:extraction}
For calibration statistics admitting the asymptotic
expansion
\begin{equation}\label{eq:expansion}
T_n = a_n\Psi(\hat{D}_n, \hat{\eta}_n) + o_p(1),
\end{equation}
where $a_n$ is a scaling sequence, $\Psi$ is a known
function, $\hat{D}_n$ is a scalar or vector-valued
statistic, and $\hat{\eta}_n$ is a vector of nuisance
estimators, the detector $\hat{D}_n$ is extracted by
the following steps:

{Step 1 (Remove scaling):} Remove the scaling sequence
$a_n$, typically $a_n = \sqrt{n}$ or $a_n = n$.

{Step 2 (Remove variance normalization):} Remove consistent
estimators of asymptotic standard deviations or variance matrices.

{Step 3 (Remove nuisance):} Remove components $\hat{\eta}_n$
whose probability limits are constant on~$\mathcal{H}_0$.

{Step 4 (Minimal consistent estimator):} The remaining
quantity $\hat{D}_n$ (scalar or vector-valued) satisfying
conditions (i) and (ii) of Definition~\ref{Def:detection} is the
detector.
\end{definition}

\begin{remark}
The extraction procedure provides a systematic heuristic for
identifying candidate detectors. Whether the extracted quantity
$\hat{D}_n$ satisfies the consistency and minimality conditions
of Definition~\ref{Def:detection} must be verified for each
specific test, typically through existing consistency results
in the literature; see Appendix~\ref{App:detector} for verification
for the four tests examined here.
\end{remark}

\begin{remark}
For the four tests examined in Section~\ref{Sect:major}, the
detector $\hat{D}_n$ and the functional $\Phi$ are as
follows:

\renewcommand{\arraystretch}{1.3}
\begin{table}[h]
\centering
\begin{tabular}{lll}
\hline
{Test} & {Detector} $\hat{D}_n$ &
  {Calibration statistic} $T_n = \Phi(\hat{D}_n, \hat\eta_n)$  \\
\hline
WMW & eAUC (scalar) &
  $\sqrt{n}(\hat{D}_n - 1/2)/\hat{\sigma}$ \\
KW &
  $(\bar{R}_1/(n+1),\ldots,\bar{R}_K/(n+1))$ (vector) &
  $\frac{12(n+1)}{n}\sum_k n_k(\hat{D}_{nk} - 1/2)^2$ \\
Friedman & $(\bar{R}_{\cdot 1},\ldots,\bar{R}_{\cdot k})$ (vector) &
  $\frac{12n}{k(k+1)}\sum_j(\hat{D}_{nj} - (k+1)/2)^2$ \\
Logrank &
  $(O-E)/n$ (scalar) &
  $\sqrt{n}\,\hat{D}_n / \sqrt{V/n}$ \\
\hline
\end{tabular}
\caption{Detectors $\hat{D}_n$ and calibration statistics $T_n = \Phi(\hat{D}_n, \hat\eta_n)$
  for the four detection-incoherent tests.}
\label{tab:detectors}
\end{table}
\renewcommand{\arraystretch}{1}

For WMW and logrank, $\Phi$ is injective (the detector
is scalar and the calibration statistic is a monotone
transformation of it). For KW and Friedman, $\Phi$ is a quadratic
form, non-injective, mapping the vector-valued detector to a
scalar. This non-injectivity of $\Phi$ does not affect minimality
of the detector: minimality concerns whether the detector
$\hat D_n$ itself can be replaced by a coarser statistic while
still recovering $\theta(P)$, whereas $\Phi$ maps $\hat D_n$ to
the calibration statistic $T_n$, whose role is calibration under
$\mathcal H_0$, not recovery of $\theta(P)$.
\end{remark}

\begin{remark}\label{Rem:monotone}
The calibration statistic is defined up to monotone rescaling:
if $T_n$ is a calibration statistic, so is any monotone
transformation $g(T_n)$, since both define identical rejection
regions. The detector, by contrast, is uniquely determined
as the minimal consistent estimator of the detection functional.
For instance, for the KW test both the calibration statistic
of Table~\ref{tab:detectors}, which we denote $H$, and
$\hat{\varepsilon}^2 = H/(n-1)$ are valid calibration statistics,
while the detector is uniquely the vector
$(\bar{R}_1/(n+1),\ldots,\bar{R}_K/(n+1))$.
\end{remark}

%
%
\begin{definition}[Detection-null value]\label{Def:nullvalue}
Assume $\theta(P) = \theta^*$ for all $P \in \mathcal{H}_0$, so
that the detection functional takes a common value over the null.
The common value $\theta^*$ is called the detection-null value.
\end{definition}

\begin{definition}[Detection-null set and detection set]
\label{Def:detectionset}
The detection-null set is
$\mathcal{D}_0 = \{P \in \mathcal{P}\mathpunct{:} \theta(P) = \theta^*\}$.
The detection set is
$\mathcal{D} = \{P \in \mathcal{P}\mathpunct{:} \theta(P) \neq \theta^*\}$.
\end{definition}

\begin{note}
The detection-null value $\theta^*$ is the sole point of contact
between $\mathcal{H}_0$ and $\mathcal{D}_0$. $\mathcal{H}_0$ is
the hypothesis under study, while $\mathcal{D}_0 =
\theta^{-1}(\{\theta^*\})$ is determined by the detector, as the
level set of~$\theta$ at the single value $\theta^*$ it takes on
$\mathcal{H}_0$. The two sets arise from different constructions,
and $\theta^*$ is the value through which they are compared.
\end{note}

%
%
\begin{definition}[Blind spot]\label{Def:blind}
The unrestricted blind spot of a test is
$$
  \mathcal{S}_{\mathcal{P}} = \mathcal{D}_0 \setminus \mathcal{H}_0,
$$
the set of distributions that belong to the detection-null set but
not the null hypothesis set. The domain-restricted blind spot is
$$
\mathcal{S}_{\mathcal{R}} =
(\mathcal{D}_0 \setminus \mathcal{H}_0) \cap \mathcal{R}.
$$
\end{definition}

\begin{proposition}[Equivalent forms of the blind spot]\label{prop:blindequiv}
Let $\mathcal{H}_0^c$ be the two-sided alternative of $\mathcal{H}_0$,
and let $\mathcal{D}$ be the two-sided complement of $\mathcal{D}_0$.
Then
$$
\mathcal{S}_{\mathcal{P}} = \mathcal{D}_0 \setminus \mathcal{H}_0
= \mathcal{H}_0^c \cap \mathcal{D}_0
= \mathcal{H}_0^c \setminus \mathcal{D}.
$$
\end{proposition}

\begin{proof}
By Definition~\ref{Def:nullvalue}, $\theta(P) = \theta^*$ for every
$P \in \mathcal{H}_0$, so $\mathcal{H}_0 \subseteq \mathcal{D}_0$.
Hence removing $\mathcal{H}_0$ from $\mathcal{D}_0$ is the same as
intersecting $\mathcal{D}_0$ with the complement of $\mathcal{H}_0$;
$\mathcal{D}_0 \setminus \mathcal{H}_0 = \mathcal{D}_0 \cap \mathcal{H}_0^c
= \mathcal{H}_0^c \cap \mathcal{D}_0$. Since $\mathcal{D}$ is the
complement of $\mathcal{D}_0$, $\mathcal{H}_0^c \cap \mathcal{D}_0
= \mathcal{H}_0^c \setminus \mathcal{D}$.
\end{proof}

\begin{corollary}[Restricted blind spot]\label{cor:blindequivR}
For any domain $\mathcal{R} \subseteq \mathcal{P}$,
$$
\mathcal{S}_{\mathcal{R}} = (\mathcal{D}_0 \setminus \mathcal{H}_0) \cap \mathcal{R}
= \mathcal{H}_0^c \cap \mathcal{D}_0 \cap \mathcal{R}
= (\mathcal{H}_0^c \setminus \mathcal{D}) \cap \mathcal{R}.
$$
\end{corollary}
\begin{proof}
Intersect each expression in Proposition~\ref{prop:blindequiv} with $\mathcal{R}$.
\end{proof}

%
Proposition~\ref{prop:blindequiv} compares two partitions of
$\mathcal{P}$: the hypothesis partition $\{\mathcal{H}_0, \mathcal{H}_0^c\}$
and the detection partition $\{\mathcal{D}_0, \mathcal{D}\}$. The
relationship between the two partitions is illustrated
in Figure~\ref{fig:partitions}: $\mathcal{H}_0 \subsetneq \mathcal{D}_0$
for a detection-incoherent test, with the blind spot
$\mathcal{S}_{\mathcal{P}} = \mathcal{D}_0 \setminus \mathcal{H}_0$.

%
%
\begin{figure}[h] 
\centering
\scalebox{0.55}{
\begin{tikzpicture}
\begin{scope}
\clip (0,0) circle (3cm);
  
\fill[gray!25] (-3,-3) rectangle (-0.7,3);

\fill[gray!10]
(-0.7,-3) rectangle (3,3);

\fill[gray!45]
(-0.7, 1.4)
.. controls (0.2, 1.0) and (0.7, 0.2) .. (0.5, -1.4)
-- (3,-3) -- (3,3) -- (-0.7,3) -- cycle;

\end{scope}

\draw[thick] (0,0) circle (3cm);

\draw[thick] (-0.7,-3) -- (-0.7,3);

\draw[thick, dashed, line width=1.2pt]
(-0.7, 1.4)
.. controls (0.2, 1.0) and (0.7, 0.2) .. (0.5, -1.4);

\draw[thick] (0,0) circle (3cm);

  
\node at (-1.7, 0.5)
{\fontsize{13}{16}\selectfont$\mathcal{H}_0$};
\node[font=\footnotesize, gray!70!black] at (-1.7, -0.1)
{\fontsize{11}{14}\selectfont$\mathcal{H}_0 \subset \mathcal{D}_0$};
%
\node[font=\normalsize] at (0.365, -1.65) 
{\fontsize{13}{16}\selectfont$\mathcal{S}_{\mathcal P}$};
\node[font=\normalsize] at (0.365, -2.05) 
{\fontsize{13}{16}\selectfont blind spot};
\node[font=\small] at (0.365, -2.5) 
{\fontsize{11}{14}\selectfont$\mathcal{D}_0  \setminus \mathcal{H}_0$};

\node[font=\Large] at (1.8, 1.8) 
{\fontsize{13}{16}\selectfont $\mathcal{D}$};

\node[font=\small] at (2.65, 2.65) 
{\fontsize{13}{16}\selectfont $\mathcal{P}$};

\node[font=\normalsize, align=center] (D0label) at (-2.15, 3.7)  
{\fontsize{13}{16}\selectfont$\mathcal{D}_0$};  

\draw[->] (-2.1, 3.42) -- (-1.6, 2.25);  

\draw[->] (-1.9, 3.42) -- (0.0, 0.4);

\node[font=\normalsize] at (1.15, -3.5)
{\fontsize{13}{16}\selectfont $\mathcal{H}_0^c$};
\draw[->] (1.15,-3.35) -- (1.15,-2.8);

\end{tikzpicture}
}
\caption{The two partitions of $\mathcal{P}$. The hypothesis
  partition $\{\mathcal{H}_0, \mathcal{H}_0^c\}$ is separated
  by the vertical line; the detection partition
  $\{\mathcal{D}_0, \mathcal{D}\}$ by the dashed arc.
  The blind spot $\mathcal{S}_{\mathcal{P}} =
    \mathcal{D}_0 \setminus \mathcal{H}_0$ (light region, right of
                                            vertical line) consists of distributions that violate
  $\mathcal{H}_0$ but produce no signal in the test statistic.
  Detection incoherence holds when
  $\mathcal{D}_0 \setminus \mathcal{H}_0 \neq \emptyset$, as shown.}
\label{fig:partitions}
\end{figure}
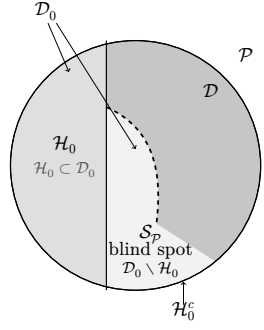

%
\begin{definition}[Detection coherence]\label{Def:calcoh}
Let $\mathcal R \subseteq \mathcal P$ be a domain. A test is
detection coherent with respect to $\mathcal R$ if its blind
spot with respect to $\mathcal R$ is empty:
$$
  \mathcal{S}_{\mathcal{R}} = \emptyset.
$$
Otherwise the test is detection incoherent with respect to
$\mathcal{R}$.

A test is (universally) detection coherent if
$\mathcal{S}_{\mathcal{P}} = \emptyset$. Otherwise the test is detection incoherent
with respect to $\mathcal{P}$.
\end{definition}

\begin{corollary}[Coherence as alternative--detection-set agreement]\label{cor:calcohequiv}
A test is detection coherent with respect to $\mathcal{R}$ if and
only if its two-sided alternative $\mathcal{H}_0^c$ coincides with its detection set $\mathcal{D}$
over~$\mathcal{R}$:
$$
  \mathcal{H}_0^c \cap \mathcal{R} = \mathcal{D} \cap \mathcal{R}.
$$
In particular, a test is universally detection coherent if and
only if $\mathcal{H}_0^c = \mathcal{D}$.
\end{corollary}

\begin{proof}
By Corollary~\ref{cor:blindequivR}, $\mathcal{S}_{\mathcal{R}}
= \mathcal{H}_0^c \cap \mathcal{D}_0 \cap \mathcal{R}$. This is
empty if and only if $\mathcal{H}_0^c \cap \mathcal{R} \subseteq
\mathcal{D} \cap \mathcal{R}$ (no point of the alternative,
restricted to $\mathcal{R}$, lies in $\mathcal{D}_0$); since
$\mathcal{D} \cap \mathcal{R} \subseteq \mathcal{H}_0^c \cap
\mathcal{R}$ always holds, this is equivalent to
$\mathcal{H}_0^c \cap \mathcal{R} = \mathcal{D} \cap \mathcal{R}$.
Taking $\mathcal{R} = \mathcal{P}$ gives the universal case.
\end{proof}

%
\begin{note}[Persistence of coherence under restrictions]
Detection coherence exhibits the following persistence properties:
  
{Coherence persists:} If an unrestricted test is detection
coherent, then it remains detection coherent with respect to any
domain $\mathcal R \subseteq \mathcal P$. Such tests are universally
detection coherent.

{Coherence can be restored:} If an unrestricted test is
detection incoherent, it may become detection coherent with
respect to appropriately chosen domains
$\mathcal R \subsetneq \mathcal P$.
\end{note}

%
%
\begin{assumption}
\label{assm:fixed-limit}
The calibration statistic $T_n$ has a fixed limiting
distribution $\mathcal{L}(P)$ at each $P \in
\mathcal{D}_0 \cap \mathcal{R}$:
$$
T_n \xrightarrow{d} \mathcal{L}(P)
\quad \text{for each } P \in \mathcal{D}_0 \cap
\mathcal{R}.
$$
\end{assumption}

\begin{proposition}[Power characterisation of the blind spot]
\label{prop:power}
Consider a test defined over domain~$\mathcal{R}$, with
$\mathcal{H}_0^c$ and $\mathcal{D}$ two-sided as in
Proposition~\ref{prop:blindequiv}.
Assume that Assumption~\ref{assm:fixed-limit} holds and
that $\theta(P) \neq \theta^*$ implies power tends 
to~$1$. Then the blind spot
$$
  \mathcal{S}_{\mathcal{R}} =
  (\mathcal{H}_0^c \setminus \mathcal{D}) \cap \mathcal{R}
$$
coincides with the set of alternatives
$P \in \mathcal{H}_0^c \cap \mathcal{R}$ against
which power converges to a fixed value less than~$1$:
$$
  \lim_{n \to \infty}
  \mathrm{P}_P(\varphi_n = 1) =
  \mathcal{L}(P)(\mathcal{C}) \in [0, 1),
$$
where $\mathcal C$ is the asymptotic critical region of the test.

Power does not grow with sample size against any
$P \in \mathcal{S}_{\mathcal{R}}$.
\end{proposition}

\begin{proof}
\textit{Forward inclusion.}
Let $P \in (\mathcal{H}_0^c \setminus \mathcal{D}) \cap \mathcal{R}$.
Then $\theta(P) = \theta^*$, so
$\hat{D}_n \xrightarrow{p} \theta^*$.
By Assumption~\ref{assm:fixed-limit},
$T_n \xrightarrow{d} \mathcal{L}(P)$ and
power converges to $\mathcal{L}(P)(\mathcal{C})
\in [0,1)$.
Hence power does not grow with sample size
against~$P$.

\textit{Reverse inclusion.}
Let $P \in \mathcal{H}_0^c \cap \mathcal{R}$
with $P \notin \mathcal{D}_0$.
Then $\theta(P) \neq \theta^*$, and by
the consistency assumption power tends to~$1$.
Hence the set of alternatives against which
power does not grow with sample size coincides
with $\mathcal{S}_{\mathcal{R}} =
\mathcal{H}_0^c \cap \mathcal{D}_0 \cap
\mathcal{R}$.
\end{proof}

%
%
\begin{corollary}[Consistency over $\mathcal{H}_0^c$]
Under the conditions of Proposition~\ref{prop:power},
if the blind spot $\mathcal{S}_{\mathcal{P}}$ is empty,
the test is consistent over all of $\mathcal{H}_0^c$.
\end{corollary}

%
%
\begin{note}[The limiting distribution at blind spot alternatives]
\label{Note:non-growing}
Proposition~\ref{prop:power} characterises the blind
spot as the set of alternatives against which power
does not grow with sample size.
At each $P^* \in \mathcal{S}_{\mathcal{R}}$,
Assumption~\ref{assm:fixed-limit} gives a fixed
limiting distribution $\mathcal{L}(P^*)$, and power
converges to
$$
c(P^*, \gamma) = \mathcal{L}(P^*)(\mathcal{C})\in [0, 1),
$$ 
a value that varies across blind spot distributions and across 
allocations $\gamma = \lim\,(n_1/n,\ldots,n_K/n)$, 
the limiting fractions of the sample falling in each group. 

The value $c(P^*, \gamma)$ depends on how
$\mathcal{L}(P^*)$ relates to the reference
distribution $\mathcal{L}_0$;
in particular: it equals~$\alpha$
when $\mathcal{L}(P^*) = \mathcal{L}_0$; exceeds~$\alpha$
when $\mathcal{L}(P^*)$ assigns more
probability to $\mathcal{C}$ than $\mathcal{L}_0$
does; falls below~$\alpha$ when $\mathcal{L}(P^*)$
assigns less; and approaches~$0$ when $\mathcal{L}(P^*)$
assigns vanishingly small probability to $\mathcal{C}$.

The WMW test (Example~\ref{ex:WMWMC}) illustrates
$c(P^*, \gamma) > \alpha$: for equal sample sizes, the null
variance formula underestimates the true variance at
gaussian heteroskedastic equal-mean blind spot distributions, giving
$c(P^*, \gamma) \approx 0.095$ at a variance ratio of $100$.
\end{note}

%
\begin{definition}[Auxiliary argument]
\label{Def:auxiliary}
The detection functional $\theta$ may depend on arguments other
than $P$; for instance, the sample allocation fraction, or a
censoring mechanism. Write $\xi \in \Xi$ for any such arguments,
so that $\theta: \mathcal{P} \times \Xi \to \Theta$. If $\theta$
depends on $P$ alone, $\Xi$ is a singleton and $\xi$ is omitted
throughout. $\xi$ is called an \emph{auxiliary argument}.
\end{definition}

\begin{definition}[Generalization to an auxiliary argument]
\label{Def:auxiliary-generalization}
If $\theta$ depends on an auxiliary argument $\xi \in \Xi$
(Definition~\ref{Def:auxiliary}), Definitions~\ref{Def:detection}
through~\ref{Def:calcoh} apply at each fixed $\xi$, with
$\theta(P)$, $\theta^*$, $\mathcal{D}_0$, $\mathcal{D}$,
$\mathcal{S}_{\mathcal{P}}$, $\mathcal{S}_{\mathcal{R}}$ read as
$\theta(P;\xi)$, $\theta^*(\xi)$, $\mathcal{D}_0(\xi)$,
$\mathcal{D}(\xi)$, $\mathcal{S}_{\mathcal{P}}(\xi)$,
$\mathcal{S}_{\mathcal{R}}(\xi)$ respectively. A test with an
auxiliary argument is detection coherent if
$\mathcal{S}_{\mathcal{R}}(\xi) = \emptyset$ at every
$\xi \in \Xi$, and detection incoherent otherwise, that is, if
$\mathcal{S}_{\mathcal{R}}(\xi) \neq \emptyset$ for at least one
$\xi \in \Xi$. The two cases are complementary. Where the
distinction is informative, a test that is coherent at some but
not all $\xi \in \Xi$ is additionally reported as coherent at
those $\xi$ at which its blind spot is empty. 
\end{definition}

%
\begin{note}[The normalization route to the detector]
\label{Note:teststat}
Definition~\ref{Def:detection} obtains the detector from the
calibration statistic $T_n$, by separating it from the scale
normalization and the nuisance that $T_n$ carries. We call this the
\emph{extraction route}, after the procedure of
Definition~\ref{Def:extraction}. A complementary route starts one step
earlier, from the test statistic $T_n^{\mathrm{test}}$ of
Definition~\ref{Def:nullhyp}, and applies a normalization instead of
removing one; we call it the \emph{normalization route}.

Suppose there is a deterministic sequence $b_n > 0$ with
$$
\frac{T_n^{\mathrm{test}}}{b_n} \xrightarrow{p} \vartheta(P),
\qquad P \in \mathcal{R},
$$
where $\vartheta$ is finite and not constant on $\mathcal{R}$. The
normalized test statistic $T_n^{\mathrm{test}}/b_n$ then satisfies the
conditions of Definition~\ref{Def:detection} -- consistency by
construction, minimality by Lemma~\ref{Lem:minimality} -- and
$\vartheta$ is its detection functional. The normalization is
determined up to a factor bounded away from zero and infinity: if
$b_n'$ is another such sequence and $b_n/b_n' \to 0$ along a
subsequence, the limit under $b_n'$ is identically zero and hence
constant, while if $b_n/b_n' \to \infty$ the limit is infinite at any
$P$ with $\vartheta(P) \neq 0$. If in addition $\vartheta$ is constant
on $\mathcal{H}_0$, Definition~\ref{Def:nullvalue} supplies the
detection-null value $\vartheta^*$ and
Definition~\ref{Def:detectionset} the detection-null set
$\{P \mathpunct{:} \vartheta(P) = \vartheta^*\}$. Where the detection
functional carries an auxiliary argument
(Definition~\ref{Def:auxiliary}), so does $\vartheta$, and the route is
read at each fixed $\xi$ as in
Definition~\ref{Def:auxiliary-generalization}.

The two routes are complementary rather than rival, and by
Proposition~\ref{prop:routes} they return the same detection-null set.
They differ in what else they deliver. The extraction route works
backward from $T_n$ and must decide which of its components are
nuisance; in return it admits the interpretable detector
-- for the KW and Friedman tests the vector of mean ranks,
which the extraction of Appendix~\ref{App:detector} selects, 
though the scalar quadratic form is admissible too.
The normalization route works forward from $T_n^{\mathrm{test}}$, which
carries no scale normalization, so there is nothing to assign; in
return it yields, for those two tests, only the scalar quadratic form
in the mean ranks.
\end{note}

\begin{proposition}[The extraction and normalization routes give the same
detection-null set]
\label{prop:routes}
Let $\hat{D}_n$ be a detector of the calibration statistic $T_n$ in the
sense of Definition~\ref{Def:detection}, with detection functional
$\theta$, attainable range $\Theta$, detection-null value $\theta^*$
(Definition~\ref{Def:nullvalue}) and detection-null set $\mathcal{D}_0$
(Definition~\ref{Def:detectionset}). Let $T_n^{\mathrm{test}}$ be the
test statistic of the same test and $b_n > 0$ deterministic, and
suppose
\begin{enumerate}[(a)]
\item \textit{Factorization:} there is a function
  $g \mathpunct{:} \Theta \to \mathbb{R}$ with
  $$
  \frac{T_n^{\mathrm{test}}}{b_n} \xrightarrow{p} g(\theta(P)),
  \qquad P \in \mathcal{R};
  $$
\item \textit{Separation:} for $\theta \in \Theta$,
  $g(\theta) = g(\theta^*)$ implies $\theta = \theta^*$.
\end{enumerate}
Then $\vartheta := g \circ \theta$ is finite on $\mathcal{R}$, takes the
common value $\vartheta^* = g(\theta^*)$ on $\mathcal{H}_0$, is not
constant on $\mathcal{R}$ whenever $\theta$ is not, and
$$
\{P \in \mathcal{R} \mathpunct{:} \vartheta(P) = \vartheta^*\}
  = \mathcal{D}_0 \cap \mathcal{R}.
$$
The blind spot, and with it the coherence verdict, is therefore the
same whichever route is taken.
\end{proposition}

\begin{proof}
Finiteness holds because $g$ is real-valued on $\Theta$, and
$\theta(P) = \theta^*$ for $P \in \mathcal{H}_0$
(Definition~\ref{Def:nullvalue}) gives
$\vartheta(P) = g(\theta^*) = \vartheta^*$ there. For the set identity,
$$
\vartheta(P) = \vartheta^*
\iff g(\theta(P)) = g(\theta^*)
\iff \theta(P) = \theta^* ,
$$
the second equivalence holding from left to right by separation and
from right to left trivially; and
$\{P \in \mathcal{R} \mathpunct{:} \theta(P) = \theta^*\}
= \mathcal{D}_0 \cap \mathcal{R}$ by
Definition~\ref{Def:detectionset}. Finally, if $\theta$ is not constant
on $\mathcal{R}$ there is $P \in \mathcal{R}$ with
$\theta(P) \neq \theta^*$, whence $\vartheta(P) \neq \vartheta^*$ by
separation, so $\vartheta$ is not constant either.
\end{proof}

\begin{remark}[The normalization route for the seven tests]
\label{Rem:routes}
Hypotheses (a) and (b) of Proposition~\ref{prop:routes} hold for every
test examined in this paper, with $b_n$, $g$ and $\vartheta^*$ as in
Table~\ref{tab:routes}.

For the WMW, Zaremba, logrank, KS and MMD tests $g$ is
the identity, which separates $\theta^*$ trivially: the normalized test
statistic is the extracted detector itself, $\vartheta = \theta$, and
$\vartheta^* = \theta^*$ is the constant that the calibration statistic
subtracts -- the $1/2$ in
$Z_{\mathrm{WMW}} = \sqrt{n}(\mathrm{eAUC} - 1/2)/\hat\sigma$, and $0$
where the test statistic is already centred, as for the logrank score
$O - E$.

For the KW and Friedman tests $g$ is a positive definite quadratic
form. It separates $\theta^*$ because it vanishes there and nowhere
else on $\Theta$ -- for KW because $\gamma_k > 0$ for every $k$, for
Friedman because the coefficients $12/\{k(k+1)\}$ are positive. Here
the centering is carried inside $g$, so $\vartheta^* = 0$, and the
normalization route returns the scalar $g(\theta)$ rather than the
vector $\theta$; the two induce the same detection-null set, as
Proposition~\ref{prop:routes} states, and the vector is retained
throughout because it exhibits the per-group and per-treatment
structure that the scalar conceals.
\end{remark}

\renewcommand{\arraystretch}{1.3}
\vspace{-\baselineskip}
\begin{table}[h]
\centering
\begin{tabular}{lllll}
\hline
{Test} & {$T_n^{\mathrm{test}}$} & {$b_n$} & {$g(\theta)$} &
  {$\vartheta^*$} \\
\hline
WMW, Zaremba &
  $U$ & $n_1 n_2$ & $\theta$ & $1/2$ \\
KW &
  $H$ & $n$ &
  $12\sum_{k} \gamma_k\bigl(\theta_k - \tfrac12\bigr)^2$ & $0$ \\
Friedman &
  $Q$ & $n$ &
  $\frac{12}{k(k+1)}\sum_{j}\bigl(\theta_j - \tfrac{k+1}{2}\bigr)^2$ &
  $0$ \\
Logrank &
  $O - E$ & $n$ & $\theta$ & $0$ \\
KS &
  $\sup_x|\hat{F}_n(x) - \hat{G}_n(x)|$ & $1$ & $\theta$ & $0$ \\
MMD &
  $\widehat{\mathrm{MMD}}^2_u$ & $1$ & $\theta$ & $0$ \\
\hline
\end{tabular}
\caption{The normalization route of Note~\ref{Note:teststat}: the test
  statistic, its normalization, and the function through which its
  probability limit factors over the detection functional $\theta$
  returned by the extraction route.}
\label{tab:routes}
\end{table}
\renewcommand{\arraystretch}{1}
\vspace{-1.25\baselineskip}

%
%
\section{Systematic Assessment of Detection Coherence}
\label{Sect:recipe}

The following procedure provides a systematic approach for
assessing the detection coherence of any statistical test.

\begin{procedure}[Assessment of detection coherence]
\label{Proc:assessment}
Given a test defined over domain $\mathcal{R} \subseteq
\mathcal{P}$:

\medskip
\noindent\textbf{Step~1.} Identify the detector $\hat{D}_n$
and the detection functional $\theta(P)$.
The detector is identified from the representation
$T_n = \Phi(\hat{D}_n, \hat{\eta}_n)$ as the component
whose convergence in probability
$$
\hat{D}_n \xrightarrow{p} \theta(P)
\quad \text{for all } P \in \mathcal{R},
$$
drives the convergence in distribution of $T_n$ under $H_0$.
For standardized statistics of the form
$T_n = a_n \hat{D}_n / \hat{\sigma}_n$,
the extraction procedure of
Definition~\ref{Def:extraction} provides a systematic
method: remove the scaling $a_n$, the variance
normalization $\hat{\sigma}_n$, and any additive nuisance
term whose probability limit does not depend on~$P$;
the probability limit of what remains is $\theta(P)$.
Alternatively the detector may be obtained from the test 
statistic $T_n^{\mathrm{test}}$ by the normalization route~(Note~\ref{Note:teststat}, 
Proposition~\ref{prop:routes}).

\medskip
\noindent\textbf{Step~2.} Obtain the detection-null value
$\theta^*$ by evaluating the detection functional at any
$P \in \mathcal{H}_0 \cap \mathcal R$:
$$
\theta^* = \theta(P), \qquad P \in \mathcal{H}_0 \cap \mathcal R.
$$

\medskip
\noindent\textbf{Step~3.} Construct the detection-null set 
$\mathcal{D}_0 = \{P \in \mathcal{P}\mathpunct{:} \theta(P) = \theta^*\}$
and its restriction $\mathcal{D}_0 \cap \mathcal{R}$ to the domain.

\medskip
\noindent\textbf{Step~4.} Check whether the null under which the test is calibrated
coincides with the detection-null set over $\mathcal{R}$:
$$
\mathcal{H}_0 \cap \mathcal{R}
\stackrel{?}{=}
\mathcal{D}_0 \cap \mathcal{R}.
$$
The test is \emph{detection coherent} with respect to
$\mathcal{R}$ if and only if the domain-restricted blind
spot
$$
\mathcal{S}_{\mathcal{R}}
= (\mathcal{D}_0\setminus\mathcal{H}_0) \cap \mathcal{R}
$$
is empty; otherwise it is \emph{detection incoherent}
with respect to $\mathcal{R}$.
\end{procedure}

\begin{note}[Procedure with an auxiliary argument]
\label{Note:recipe-auxiliary}
The steps above are written for a detection functional $\theta(P)$.
When the probability limit in Step~1 depends on an auxiliary argument $\xi$ as well as on~$P$, 
the functional is $\theta(P;\xi)$ (Definition~\ref{Def:auxiliary}), and Steps~1--3 are to be read 
at fixed $\xi \in \Xi$: Step~2 requires
Definition~\ref{Def:nullvalue} to hold at each $\xi$ separately and
returns $\theta^*(\xi)$, Step~3 returns a family
$\mathcal{D}_0(\xi)$, and the blind spot becomes
$\mathcal{S}_{\mathcal{R}}(\xi)$.
 
Step~4 changes character. The verdict is not obtained at a single
$\xi$: by Definition~\ref{Def:auxiliary-generalization} the test is
detection coherent with respect to $\mathcal{R}$ only if
$\mathcal{S}_{\mathcal{R}}(\xi) = \emptyset$ at \emph{every}
$\xi \in \Xi$, and detection incoherent as soon as
$\mathcal{S}_{\mathcal{R}}(\xi) \neq \emptyset$ for one. 
Checking a single design 
establishes neither
verdict.
 
Whether $\Xi$ is a singleton is settled in Step~1: it is precisely when the probability limit of the 
detector depends on $P$ alone, in which case the procedure reduces to the form stated above. Of the tests 
examined here, the KW and logrank detection functionals carry the allocation~$\gamma$ as auxiliary argument, 
and the WMW, Friedman, Zaremba, KS and MMD detection functionals carry none
(Appendices~\ref{App:WMW} and~\ref{App:Friedman}; Note~\ref{Note:two-aux}).
\end{note}

\begin{note}[Minimality of the detector]
\label{Note:minimality}
The detector $\hat{D}_n$ identified in Step~1 should satisfy
the minimality condition of Definition~\ref{Def:detection}(ii).
For the tests considered in this paper
(WMW, KW, Friedman, logrank, KS, Zaremba, MMD)
minimality is satisfied automatically.
The condition is relevant when multiple representations
of the calibration statistic are available; in that case
the minimal one should be used, as a coarser representation
may obscure the detection-null set.
\end{note}

\begin{note}[Step~2 is well-defined]
\label{Note:step2-welldefined}
By Definition~\ref{Def:nullvalue}, $\theta^*$ does not depend
on which $P\in\mathcal{H}_0$ is used to compute it. In
practice the most convenient choice is typically the simplest
member of the null, such as equal distributions or a specific
parametric model.
\end{note}
%
%
%
%
\section{Detection Incoherence of Major Unrestricted Nonparametric
Tests}\label{Sect:major}

We demonstrate that four fundamental nonparametric tests (the
WMW test, KW test, Friedman test, and logrank test) are detection
incoherent when applied as unrestricted tests.

For each unrestricted test, we present: null hypothesis, test
statistic, calibration statistic, detector, detection
functional, detection-null value, detection-null set, blind spot,
and detection incoherence; the incoherence finding is then recorded
as a proposition, proved by a constructive example demonstrating
that $\mathcal{S}_{\mathcal P} \neq \emptyset$,
or that $\mathcal{S}_{\mathcal P}(\gamma) \neq \emptyset$ for every
$\gamma$, where the detection functional depends on the
allocation $\gamma$.

%
\subsection{Wilcoxon--Mann--Whitney test}

\begin{description}
\item{Null hypothesis}: $\mathcal H_0 = \{(F_1,F_2)\mathpunct{:} F_1 = F_2\}$;
cf.~\cite{wilcoxon1945individual, mann1947test,
lehmann1975nonparametrics, lee1990u, maritz1995distribution}
\item{Test statistic}: $U = \sum_{i=1}^{n_1}\sum_{j=1}^{n_2}
    \Bigl[{1}\{X_i < Y_j\} + \tfrac12{1}\{X_i = Y_j\}\Bigr]$
\item{Calibration statistic}: $Z_{\mathrm{WMW}} = \sqrt{n}(\mathrm{eAUC} - 1/2)/\hat{\sigma}$ with $N(0,1)$ reference distribution under~$H_0$
\item{Detector}: $\mathrm{eAUC} = U/(n_1 n_2)$
\item{Detection functional}: $\theta(F_1, F_2) = \mathrm{AUC} = A_{12} = 
\mathrm{P}(X < Y) + \tfrac{1}{2}\mathrm{P}(X = Y)$
\item{Detection-null value}: $\theta^* = 1/2$
\item{Detection-null set}: $\mathcal{D}_0 = \left\{(F_1,F_2): \mathrm{AUC} = 1/2\right\}$
\item{Blind spot}: $\mathcal{S}_{\mathcal P} = \mathcal{D}_0
\setminus \mathcal{H}_0 = \left\{(F_1,F_2)\mathpunct{:} \mathrm{AUC} = 1/2\right\}
\setminus \left\{(F_1,F_2)\mathpunct{:} F_1 = F_2\right\}$
\item{Detection incoherence}: $\mathcal{S}_{\mathcal P} \neq \emptyset$
\end{description}

\begin{proposition}[Detection incoherence of the unrestricted WMW test]
\label{prop:wmw-incoherent}
The unrestricted Wilcoxon--Mann--Whitney test of
$\mathcal{H}_0 = \{(F_1,F_2)\mathpunct{:} F_1 = F_2\}$ is detection
incoherent: $\mathcal{S}_{\mathcal{P}} \neq \emptyset$.
\end{proposition}

\begin{proof}
Example~\ref{ex:WMWMC} exhibits a member of $\mathcal{S}_{\mathcal{P}}$.
\end{proof}

\begin{example}\label{ex:WMWMC}
Let $X_1 \sim N(0,1)$ and $X_2 \sim N(0,\sigma^2)$
for any $\sigma \neq 1$. Both distributions are symmetric about
the same point, so $X_1 - X_2$ is symmetric about $0$ and
$\mathrm{AUC} = 1/2$ exactly; hence $(F_1,F_2) \in \mathcal{D}_0$.
However, $F_1 \neq F_2$ whenever
$\sigma \neq 1$, so $(F_1,F_2) \notin \mathcal{H}_0$. Hence
$(F_1,F_2) \in \mathcal D_0 \setminus \mathcal H_0 = \mathcal{S}_{\mathcal{P}}$,
demonstrating $\mathcal{S}_{\mathcal P} \neq \emptyset$.

A Monte Carlo study (see~Figure~\ref{fig:WMWMC})
demonstrates that WMW power against the blind spot alternative
($N(0,1)$ and $N(0,100)$)
converges for equal group sizes to approximately $0.095$ and remains there regardless of~$n$,
while power against an alternative outside the blind spot
($N(0,1)$ and $N(0.5,1)$) grows with~$n$.

\begin{figure}[htbp] 
\centering
    \includegraphics[width=0.5\textwidth]{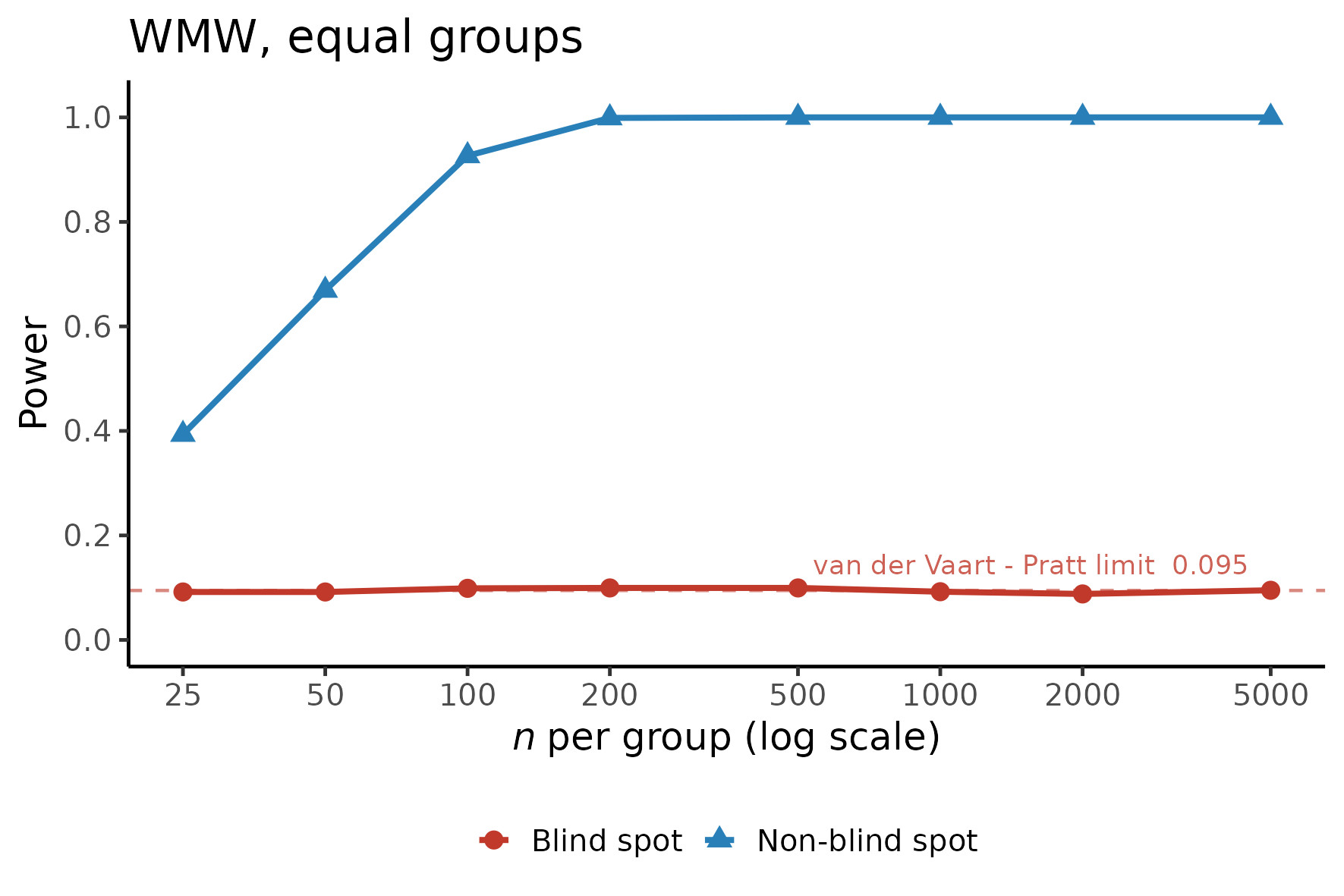}
    \caption{Power of the WMW test against two  alternatives: the blind
    spot distribution ($N(0,1)$ and $N(0,100)$, $\mathrm{AUC} = 0.5$, red circles) and a near-blind-spot distribution
    ($N(0,1)$ and $N(0.5,1)$, $\mathrm{AUC} \neq 0.5$, blue triangles).
    For equal group size, power against the blind spot alternative converges to approximately $0.095$ and does not grow with~$n$;
    power against the non-blind-spot alternative grows with~$n$,
    reaching $1$ at $n \approx 200$. Dashed line: the asymptotic rejection
    rate $0.095$ at this distribution pair and equal allocation
    (van der Vaart~\cite{vandervaart1961}, Pratt~\cite{pratt1964robustness});
    it is specific to the variance ratio $100$ and rises to $0.110$ as
    the ratio grows without bound.}
    \label{fig:WMWMC}
\end{figure}

Van der Vaart~\cite{vandervaart1961} derived the asymptotic rejection rate
of the WMW test at heteroskedastic equal-mean normal distributions as a function
of the variance ratio and sample size ratio, classifying it as a probability of Type I error;
Pratt~\cite{pratt1964robustness} computed related quantities by a different method.
Since these distributions lie in $\mathcal{H}_0^c$, the rejection rate is power,
not Type I error. The confusion arises because the DGP van der Vaart and Pratt study
happens to fall in WMW's blind spot. At blind spot distributions, power does not grow
with sample size -- behaviorally indistinguishable from Type~I error inflation
if the existence of blind spot is unrecognized.
Detection coherence is conceptually distinct from the robustness of
tests~(cf.~Box~\cite{box1953}): robustness asks whether Type~I error
is controlled when distributional assumptions are violated; detection
coherence asks whether the test has a blind spot.

\end{example}

%
\subsection{Kruskal--Wallis test}

\begin{description}
\item{Null hypothesis}: $\mathcal H_0 = \{(F_1,\ldots,F_K)\mathpunct{:}
  F_1 = \cdots = F_K\}$; cf.~\cite{kruskal1952use}

\item{Test statistic}: $H = \frac{12}{n(n+1)} \sum_{k=1}^K n_k
  \left(\bar{R}_k - \frac{n+1}{2}\right)^2$,
  where $R_{ki}$ is the rank of the $i$-th observation of group~$k$
  in the pooled sample of size $n = \sum_k n_k$ and
  $\bar R_k = n_k^{-1}\sum_i R_{ki}$ is the mean rank of group~$k$

\item{Calibration statistic}: $H$ with $\chi^2_{K-1}$ reference
  distribution under $\mathcal{H}_0$

\item{Detector}: $(\bar{R}_1/(n+1),\ldots,\bar{R}_K/(n+1))$. Writing
  the rank as a count of comparisons gives
  $$
  \bar{R}_k = 1 + \frac{n_k-1}{2}
            + \sum_{r \neq k} n_r\,\hat{A}_{rk},
  \qquad
  \hat{A}_{rk} = \frac{1}{n_r n_k}\sum_{s=1}^{n_r}\sum_{i=1}^{n_k}
    \Bigl[{1}\{X_{rs} < X_{ki}\}
        + \tfrac12{1}\{X_{rs} = X_{ki}\}\Bigr],
  $$
  the empirical AUC of group~$r$ against group~$k$, so that with
  $\hat\gamma_r = n_r/n$ and the convention $\hat{A}_{kk} = 1/2$,
  $$
  \frac{\bar{R}_k}{n+1}
    = \sum_{r=1}^{K} \hat\gamma_r\,\hat{A}_{rk} + O(n^{-1}) :
  $$
  each component of the detector is a $\hat\gamma$-weighted average of
  empirical AUCs against group~$k$

\item{Detection functional}: $\theta(F_1,\ldots,F_K;\gamma) =
  (A_1(\gamma),\ldots,A_K(\gamma))$, where
  $$
  A_k(\gamma) = \sum_{r=1}^K \gamma_r A_{rk}
              = \mathrm{P}(X_\gamma < X_k)
              + \tfrac12\,\mathrm{P}(X_\gamma = X_k)
  $$
  is the AUC of the pooled mixture against group~$k$;
  $A_{rk} = \mathrm{P}(X_r < X_k) + \tfrac12\mathrm{P}(X_r = X_k)$ is the
  pairwise AUC of group~$r$ against group~$k$, $X_g \sim F_g$
  independent, $X_\gamma \sim \sum_{k=1}^K \gamma_k F_k$,
  $\gamma = (\gamma_1,\ldots,\gamma_K)$, $\gamma_k = \lim n_k/n$

\item{Detection-null value}: $\theta^*(\gamma) = (1/2,\ldots,1/2)$

\item{Detection-null set}: $\mathcal{D}_0(\gamma) =
  \{(F_1,\ldots,F_K)\mathpunct{:} (A_1(\gamma),\ldots,A_K(\gamma))
  = (1/2,\ldots,1/2)\}$

\item{Blind spot}: $\mathcal{S}_{\mathcal P}(\gamma) =
  \mathcal{D}_0(\gamma) \setminus \mathcal{H}_0$

\item{Detection incoherence}: $\mathcal{S}_{\mathcal P}(\gamma)
  \neq \emptyset$ for every~$\gamma$
\end{description}

\begin{note}
The KW test is among tests for which Brunner et al.~\cite{brunner2021}
demonstrated dependence of the consistency region on the allocation $\gamma$.
In the detection coherence framework this implies the detection functional 
and detection-null set of KW depending on $\gamma$
(Note~\ref{Note:auxiliary}, Definition~\ref{Def:auxiliary-generalization}); 
see also~Appendix~\ref{App:KW}.
\end{note}

\begin{proposition}[Detection incoherence of the unrestricted KW test]
\label{prop:kw-incoherent}
The unrestricted Kruskal--Wallis test of
$\mathcal{H}_0 = \{(F_1,\ldots,F_K)\mathpunct{:} F_1 = \cdots = F_K\}$
is detection incoherent at every allocation:
$\mathcal{S}_{\mathcal{P}}(\gamma) \neq \emptyset$ for every
$\gamma$ on the simplex.
\end{proposition}

\begin{proof}
Example~\ref{ex:kw} exhibits a single distribution triple lying in
$\mathcal{S}_{\mathcal{P}}(\gamma)$ simultaneously for every~$\gamma$.
\end{proof}

\begin{example}\label{ex:kw}
Consider three normal distributions: $X_1 \sim N(0,1)$, $X_2 \sim N(0,4)$, 
$X_3 \sim N(0,9)$, and let $\gamma$ be any allocation on the simplex. 
Each $F_k$ is symmetric about $0$, so the pooled mixture $\sum_k \gamma_k F_k$
is symmetric about $0$ for \emph{any} mixing weights, and $X_\gamma - X_k$ 
is symmetric about $0$ as a difference of independent variables symmetric about $0$. 
Hence $A_k(\gamma) = \mathrm{P}(X_\gamma < X_k) = 1/2$ for every $k$ and every $\gamma$,
the tie term vanishing by continuity, so the detection functional takes its detection-null value, 
$\theta(P;\gamma) = (1/2,1/2,1/2) = \theta^*(\gamma)$, 
and $(F_1,F_2,F_3) \in \mathcal{D}_0(\gamma)$. 
Yet $F_1, F_2, F_3$ are distinct, so $(F_1,F_2,F_3) \notin \mathcal{H}_0$. 
Hence $(F_1,F_2,F_3) \in \mathcal{D}_0(\gamma) \setminus \mathcal{H}_0 = \mathcal{S}_{\mathcal{P}}(\gamma)$ at every allocation.
\end{example}

%
\subsection{Friedman test}
 
\begin{description}
\item{Null hypothesis}: $\mathcal H_0 = \{P\mathpunct{:} F_{i1} = \cdots =
  F_{ik}\ \forall\, i = 1,\ldots,n\}$ (equal within-block treatment
  distributions); cf.~\cite{friedman1937use}
 
\item{Test statistic}: $Q = \frac{12n}{k(k+1)} \sum_{j=1}^k
  \left(\bar{R}_{\cdot j} - \frac{k+1}{2}\right)^2$,
  where $R_{ij}$ is the rank of treatment~$j$ within block~$i$, taken
  among the $k$ observations of that block, and
  $\bar R_{\cdot j} = n^{-1}\sum_{i=1}^n R_{ij}$
 
\item{Calibration statistic}: $Q$ with $\chi^2_{k-1}$ reference
  distribution under $\mathcal{H}_0$
 
\item{Detector}: $(\bar{R}_{\cdot 1},\ldots,\bar{R}_{\cdot k})$.
  Writing the within-block rank as a count of comparisons gives
  $$
  \bar{R}_{\cdot j} = 1 + \sum_{l \neq j} \hat{A}_{lj},
  \qquad
  \hat{A}_{lj} = \frac1n \sum_{i=1}^n
    \Bigl[{1}\{X_{il} < X_{ij}\}
        + \tfrac12{1}\{X_{il} = X_{ij}\}\Bigr],
  $$
  the block-averaged empirical AUC of treatment~$l$ against
  treatment~$j$: each component of the detector is a sum of $k-1$
  within-block empirical AUCs
 
\item{Detection functional}: $\theta(P) = (\mu_1,\ldots,\mu_k)$, where
  $\mu_j = 1 + \sum_{l \neq j} \bar{A}_{lj}$ and
  $$
  \bar{A}_{lj} = \lim_{n\to\infty} \frac1n \sum_{i=1}^n
    \Bigl[ \mathrm{P}(X_{il} < X_{ij})
         + \tfrac12\,\mathrm{P}(X_{il} = X_{ij}) \Bigr]
  $$
  is the block-averaged AUC of treatment~$l$ against treatment~$j$,
  satisfying $\bar{A}_{lj} + \bar{A}_{jl} = 1$. Unlike the pairwise
  AUCs $A_{rk}$ of the KW test, the two variables compared are
  coordinates of one block and need not be independent
 
\item{Detection-null value}: $\theta^* = ((k+1)/2,\ldots,(k+1)/2)$
 
\item{Detection-null set}: $\mathcal{D}_0 = \{P\mathpunct{:}
  (\mu_1,\ldots,\mu_k) = ((k+1)/2,\ldots,(k+1)/2)\}$
 
\item{Blind spot}: $\mathcal{S}_{\mathcal P} = \mathcal{D}_0
  \setminus \mathcal{H}_0$
 
\item{Detection incoherence}: $\mathcal{S}_{\mathcal P} \neq \emptyset$
\end{description}

\begin{proposition}[Detection incoherence of the unrestricted Friedman test]
\label{prop:friedman-incoherent}
The unrestricted Friedman test of
$\mathcal{H}_0 = \{P\mathpunct{:} F_{i1} = \cdots = F_{ik}\ \forall\, i\}$
is detection incoherent: $\mathcal{S}_{\mathcal{P}} \neq \emptyset$.
\end{proposition}

\begin{proof}
Example~\ref{ex:friedman} exhibits a member of $\mathcal S_{\mathcal P}$.
\end{proof}

%
\begin{example}\label{ex:friedman}
Consider $k = 3$ treatments observed on $n$ subjects, with
$$
X_{ij} = b_i + \sigma_j\,\varepsilon_{ij},
\qquad \sigma = (\sigma_1,\sigma_2,\sigma_3) = (1,2,3),
$$
where $b_i$ is a subject effect and the $\varepsilon_{ij}$ are
i.i.d.\ symmetric about~$0$. The treatments differ in scale only.
The observations within a subject are correlated through $b_i$,
$$
\mathrm{corr}(X_{il}, X_{ij}) =
\frac{\mathrm{var}(b)}
     {\sqrt{\{\mathrm{var}(b)+\sigma_l^2\}\{\mathrm{var}(b)+\sigma_j^2\}}}.
$$

Each within-block difference,
$$
X_{il} - X_{ij} = \sigma_l\varepsilon_{il} - \sigma_j\varepsilon_{ij},
$$
is symmetric about~$0$, the subject effect having cancelled. Hence
every pairwise AUC equals one half,
$\bar{A}_{lj} = \mathrm{P}(X_{il} < X_{ij}) = 1/2$ for $l \neq j$,
and the detection functional takes its detection-null value,
$$
\mu_j = 1 + \sum_{l \neq j} \bar{A}_{lj}
      = 1 + \frac{k-1}{2} = \frac{k+1}{2} = 2,
\qquad j = 1,2,3,
$$
so $\theta(P) = (2,2,2) = \theta^*$ and $P \in \mathcal{D}_0$. Yet
$F_{i1}, F_{i2}, F_{i3}$ differ in scale in every block, so
$P \notin \mathcal{H}_0$. Hence
$P \in \mathcal{D}_0 \setminus \mathcal{H}_0 =
\mathcal{S}_{\mathcal{P}}$, demonstrating
$\mathcal{S}_{\mathcal{P}} \neq \emptyset$.

\end{example}

%
\subsection{Logrank test}\label{Subsect:logrank}

\begin{description}
\item{Null hypothesis}: $\mathcal H_0 = \{(S_1,S_2)\mathpunct{:} S_1(t) =
S_2(t) \text{ for all } t\}$;
cf.~\cite{mantel1966evaluation, peto1972asymptotically}
\item{Test statistic}: $U_n = O - E = \sum_j(d_{1j} - E_{1j})$
  \item{Calibration statistic}: $Z = U_n/\sqrt{V_n}$  with $N(0,1)$ reference distribution under $H_0$
  \item{Detector}: $U_n/n$
  \item{Detection functional}: $\theta(S_1,S_2;\gamma) =
  \int_0^\infty w(t;\gamma)[\lambda_1(t) - \lambda_2(t)]\,dt$
  where $w(t;\gamma) = \gamma(1-\gamma)\,y_1(t)y_2(t) /
  \{\gamma\,y_1(t) + (1-\gamma)\,y_2(t)\}$, $y_i(t) = S_i(t)$,
$\gamma = \lim n_1/n$
  \item{Detection-null value}: $\theta^*(\gamma) = 0$
  \item{Detection-null set}: $\mathcal{D}_0(\gamma) = \left\{(S_1,S_2)\mathpunct{:}
    \int_0^\infty w(t;\gamma)[\lambda_1(t) - \lambda_2(t)]\,dt
    = 0\right\}$
  \item{Blind spot}: $\mathcal{S}_{\mathcal P}(\gamma) =
  \mathcal{D}_0(\gamma) \setminus \mathcal{H}_0$
  \item{Detection incoherence}: $\mathcal{S}_{\mathcal P}(\gamma)
\neq \emptyset$ for every $\gamma$
\end{description}

\begin{note}
The allocation-dependence (Note~\ref{Note:auxiliary}) arises
for the logrank test: the weight $w(t;\gamma)$ pools the two arms
with allocation-dependent weights, so the allocation $\gamma$
remains in the detection functional. The detection functional and
detection-null set are consequently indexed by the auxiliary
argument $\gamma$ (Definition~\ref{Def:auxiliary-generalization}).
\end{note}

\begin{lemma}[Crossing-hazard pairs in the blind spot at every allocation]
\label{lem:cstar}
Let $a \neq b$ be positive and, for $c > 0$, let $T_1$ and $T_2$
have piecewise exponential distributions with hazard functions
$$
\lambda_1(t;c) = \begin{cases} a & t \leq c \\ b & t > c \end{cases},
\qquad
\lambda_2(t;c) = \begin{cases} b & t \leq c \\ a & t > c \end{cases},
$$
and survival functions
$S_i(t;c) = \exp\{-\int_0^t \lambda_i(s;c)\,ds\}$. Then for every
$\gamma \in (0,1)$ there exists $c^*(\gamma) \in (0,\infty)$ such
that $\theta(P;\gamma) = 0$ at $c = c^*(\gamma)$, and the resulting
pair lies in $\mathcal{S}_{\mathcal{P}}(\gamma)$.
\end{lemma}

\begin{proof}
Fix $\gamma \in (0,1)$ and write $y_i(t;c) = S_i(t;c)$.
Evaluating the detection functional $\theta(S_1,S_2;\gamma)$ 
along the family gives $\theta(c;\gamma) := \theta(S_1(\cdot\,;c), S_2(\cdot\,;c);\gamma) = 
\int_0^\infty w(t;c,\gamma)[\lambda_1(t;c) - \lambda_2(t;c)]\,dt$, 
where 
$$
w(t;c,\gamma) = \frac{\gamma(1-\gamma)\,y_1(t;c)\,y_2(t;c)}
                     {\gamma y_1(t;c) + (1-\gamma)y_2(t;c)}.
$$

\textit{Continuity.} For all finite $t,c$ one has
$S_i(t;c) \in (0,1]$, so the denominator is strictly positive and
$w(t;c,\gamma) > 0$. Since
$\gamma y_1 + (1-\gamma)y_2 \geq \gamma y_1$ and likewise
$\geq (1-\gamma) y_2$,
$$
w(t;c,\gamma) \;\leq\; \min\{\gamma\, y_1(t;c),\,(1-\gamma)\,y_2(t;c)\}
\;\leq\; e^{-\min(a,b)\,t},
$$
the last step because $\int_0^t \lambda_i(s;c)\,ds \geq \min(a,b)\,t$.
The integrand is thus dominated by $|a-b|\,e^{-\min(a,b)t}$,
uniformly in $c$, and is continuous in $c$ for each~$t$. By
dominated convergence $\theta(\cdot\,;\gamma)$ is continuous on
$(0,\infty)$.

\textit{Boundary limits.} As $c \to 0^+$, $\lambda_1 \to b$ and
$\lambda_2 \to a$ for almost every $t$, and
$S_1(t;c) \to e^{-bt}$, $S_2(t;c) \to e^{-at}$; the same domination
gives
$$
\theta(0^+;\gamma) = (b-a)\int_0^\infty w_0(t;\gamma)\,dt,
\qquad
w_0(t;\gamma) = \frac{\gamma(1-\gamma)e^{-(a+b)t}}
                     {\gamma e^{-bt} + (1-\gamma)e^{-at}} > 0 .
$$
As $c \to \infty$, for each fixed $t$ eventually $c > t$, so
$\lambda_1 = a$, $\lambda_2 = b$ and $S_1(t;c) \to e^{-at}$,
$S_2(t;c) \to e^{-bt}$, giving
$$
\theta(\infty;\gamma) = (a-b)\int_0^\infty w_\infty(t;\gamma)\,dt ,
\qquad
w_\infty(t;\gamma) = \frac{\gamma(1-\gamma)e^{-(a+b)t}}
                          {\gamma e^{-at} + (1-\gamma)e^{-bt}} > 0 .
$$
Both integrals are finite and strictly positive, so the two limits
are non-zero and of opposite sign.

\textit{Existence.} By the intermediate value theorem there is
$c^*(\gamma) \in (0,\infty)$ with $\theta(c^*(\gamma);\gamma) = 0$,
so the pair lies in $\mathcal{D}_0(\gamma)$. Since $a \neq b$ the
hazards cross at $c^*(\gamma)$ and $S_1 \neq S_2$, so the pair is
not in $\mathcal{H}_0$. Hence it lies in
$\mathcal{D}_0(\gamma)\setminus\mathcal{H}_0 =
\mathcal{S}_{\mathcal{P}}(\gamma)$.
\end{proof}

\begin{proposition}[Detection incoherence of the unrestricted logrank test]
\label{prop:logrank-incoherent}
The unrestricted logrank test of
$\mathcal{H}_0 = \{(S_1,S_2)\mathpunct{:} S_1 = S_2\}$ is detection
incoherent at every allocation:
$\mathcal{S}_{\mathcal{P}}(\gamma) \neq \emptyset$ for every
$\gamma \in (0,1)$.
\end{proposition}

\begin{proof}
Immediate from Lemma~\ref{lem:cstar}.
\end{proof}

\begin{example}\label{ex:logrankMC}
Take $a = 0.5$ and $b = 2.0$ in Lemma~\ref{lem:cstar}. Solving
$\theta(c;\gamma) = 0$ numerically gives $c^*(0.3) \approx 0.485$,
$c^*(0.5) \approx 0.518$ and $c^*(0.7) \approx 0.563$: at each
allocation the corresponding crossing-hazard pair has
$\theta(P;\gamma) = 0 = \theta^*(\gamma)$ while $S_1 \neq S_2$, and
so lies in the blind spot $\mathcal{S}_{\mathcal{P}}(\gamma)$.

\begin{figure}[htbp]
\centering
\includegraphics[width=0.75\textwidth]{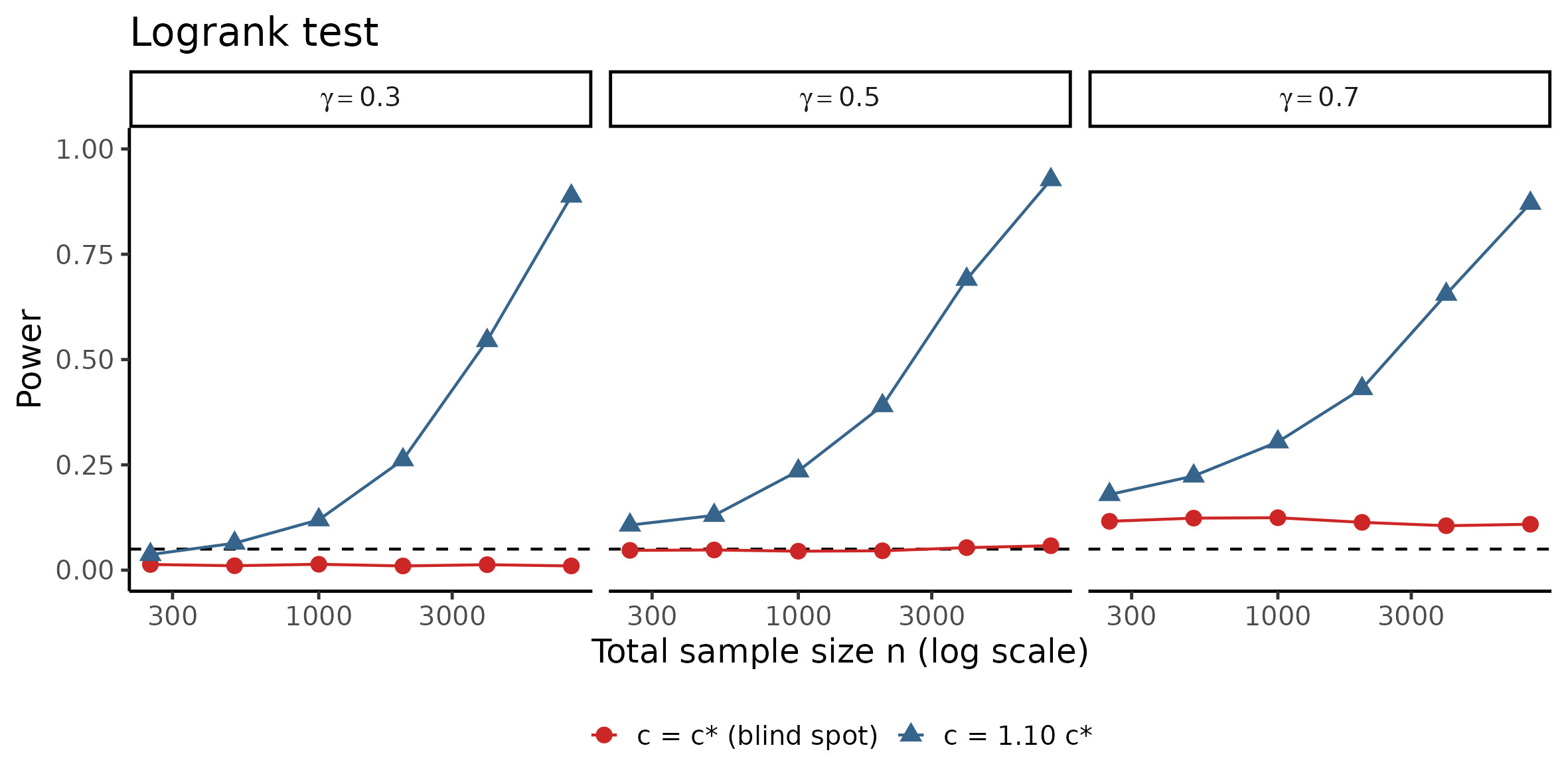}
\caption{Monte Carlo power of the logrank test at three allocations,
  $\gamma \in \{0.3, 0.5, 0.7\}$. Within each panel the two
  piecewise exponential hazard pairs share the rates $a = 0.5$,
  $b = 2.0$ and differ only in the breakpoint: the blind spot pair
  $c = c^*(\gamma)$, for which $\theta(P;\gamma) = 0$, and a
  perturbed pair $c = 1.1\,c^*(\gamma)$, for which
  $\theta(P;\gamma) \neq 0$. Both have $S_1 \neq S_2$ with crossing
  hazards. Power against the blind spot pair stays at a fixed level
  as $n$ grows, a level that itself depends on the allocation; power
  against the perturbed pair grows towards~$1$ at every allocation.
  Dashed line: nominal level $\alpha = 0.05$. The two are
  empirically indistinguishable at small $n$; the permanent nature
  of the blind spot power failure is revealed only by varying~$n$.}
\label{fig:MC}
\end{figure}

A Monte Carlo study (Figure~\ref{fig:MC}) contrasts, at each of the
three allocations, the blind spot pair with breakpoint $c^*(\gamma)$
against a perturbed pair with breakpoint $1.1\,c^*(\gamma)$. Power
against the blind spot pair does not grow with the sample size~$n$;
power against the perturbed pair grows with~$n$. Both configurations
have $S_1 \neq S_2$ with crossing hazards, and the difference in
power behavior is entirely attributable to whether
$\theta(P;\gamma) = 0$ or not.

Power of the logrank close to the nominal level was noted in Monte
Carlo studies of realistic clinical trial scenarios
(cf.~\cite{cheng2021}).
\end{example}

\subsection{Summary}

Table~\ref{tab:systematic} summarizes the null hypothesis, detector
and detection-null set for each of the four tests;
by Propositions~\ref{prop:wmw-incoherent},
\ref{prop:kw-incoherent}, \ref{prop:friedman-incoherent}
and~\ref{prop:logrank-incoherent} all four are detection incoherent
as unrestricted tests: their null hypotheses~$\mathcal H_0$ are
strictly smaller than their detection-null sets~$\mathcal D_0$,
creating non-empty blind spots.

\renewcommand{\arraystretch}{1.15}

\begin{table}[h]
\centering
\smallskip
\begin{tabular}{llll}
\hline
{Test} &
{Null Hypothesis}  $\mathcal H_0$ &
{Detector} $\hat{D}_n$ &
{Detection-Null Set} $\mathcal D_0$ \\
\hline
WMW &
  $F_1 = F_2$ &
  eAUC &
  $\{\mathrm{AUC} = 1/2\}$  \\
KW &
  $F_1 = \cdots = F_K$ &
  $\bigl(\bar{R}_1/(n+1),\ldots,\bar{R}_K/(n+1)\bigr)$ &
  $\{A_k(\gamma) = 1/2,\ k = 1,\ldots,K\}^{\,\ddagger}$  \\
Friedman &
  $F_{i1} = \cdots = F_{ik}\ \forall\, i$ &
  $\bigl(\bar{R}_{\cdot 1},\ldots,\bar{R}_{\cdot k}\bigr)$ &
  $\{\mu_j = (k+1)/2,\ j = 1,\ldots,k\}$ \\
Logrank &
  $S_1 = S_2$ &
  $(O-E)/n$ &
  $\bigl\{\int w(t;\gamma)[\lambda_1(t)-\lambda_2(t)]\,dt
    = 0\bigr\}^{\,\ddagger}$ \\
\hline
\end{tabular}
\begin{minipage}{\linewidth}
\smallskip
\footnotesize{$^\ddagger$ For KW and logrank the detection functional, and hence
the detection-null set, is indexed by the allocation $\gamma$
(Note~\ref{Note:auxiliary}, Definition~\ref{Def:auxiliary-generalization});
these entries are to be read at each fixed~$\gamma$, with
$A_k(\gamma) = \mathrm{P}(X_\gamma < X_k) + \tfrac12\,\mathrm{P}(X_\gamma = X_k)$
the AUC of the pooled mixture $X_\gamma \sim \sum_r \gamma_r F_r$ against
group~$k$, and $w(t;\gamma) =
\gamma(1-\gamma)y_1y_2/\{\gamma y_1 + (1-\gamma)y_2\}$. The Friedman
entry uses $\mu_j = 1 + \sum_{l \neq j}\bar{A}_{lj}$, the mean
within-block rank of treatment~$j$.}
\end{minipage}
\caption{Detection coherence assessment of four fundamental
nonparametric tests.}
\label{tab:systematic}
\end{table}
\renewcommand{\arraystretch}{1}
\vspace{-1\baselineskip}

The paragraphs below describe each blind spot in the equivalent
two-sided form $\mathcal{H}_0^c \setminus \mathcal{D}$
(Proposition~\ref{prop:blindequiv}), contrasting the full
alternative $\mathcal{H}_0^c$ with the detection set $\mathcal D$ --
what each test's detection functional can distinguish from the null.

The WMW test is calibrated under $H_0\mathpunct{:} F_1=F_2$, implying it 
should detect any distributional difference. Its detection functional, 
$\theta(P) = \mathrm{AUC} = A_{12}$, estimated by the detector, eAUC, compares 
the two samples directly with each other, and is sensitive only to departures in 
which one group is stochastically favored over the other, $\mathrm{AUC} \neq 1/2$. 
The test is blind to distributional differences that preserve $\mathrm{AUC} = 1/2$: 
variance differences, shape differences, and bimodal versus unimodal configurations

The KW test is calibrated under $H_0\mathpunct{:} F_1=\cdots=F_K$, implying
it should detect any difference across $K$ groups. Its detection functional,
$\theta(P;\gamma) = (A_1(\gamma),\ldots,A_K(\gamma))$ with
$A_k(\gamma) = \mathrm{P}(X_\gamma < X_k) + \tfrac{1}{2}\mathrm{P}(X_\gamma = X_k)$, 
estimated by the detector, the vector of
mean within-group ranks, is sensitive only to departures in which some group
is stochastically favored or disfavored against the $\gamma$-weighted
pooled mixture $X_\gamma$. The test is blind to differences that leave every group at $A_k(\gamma) = 1/2$ 
against that mixture -- groups symmetric about a common point and differing only 
in variance or shape, for instance Example~\ref{ex:kw}.
Because the reference is the pooled
mixture rather than the groups themselves, both what the test detects and
what it misses move with the allocation~$\gamma$.

The Friedman test is calibrated under $H_0\mathpunct{:} F_{i1} = \cdots = F_{ik}$ 
for all blocks~$i$, implying it should detect any within-block distributional difference 
between treatments. Its detection functional is the vector of expected within-block ranks, 
whose $j$-th component $\mu_j = 1 + \sum_{l \neq j} \bar{A}_{lj}$ aggregates the block-averaged 
AUCs of treatment~$j$ against the others. It reaches its detection-null value precisely 
when every treatment has average AUC $1/2$ against the other treatments. The test is blind 
to any configuration achieving that balance while the within-block distributions differ. 
Treatments differing in scale about a common centre (Example~\ref{ex:friedman}) 
is one such configuration.

The logrank test is calibrated under $H_0\mathpunct{:} S_1=S_2$, implying it
should detect any survival difference. Its detection functional,
$\theta(P;\gamma) = \int w(t;\gamma)(\lambda_1(t)-\lambda_2(t))\,dt$,
estimated by the detector, $(O-E)/n$, is sensitive only to survival
differences where this integral is non-zero. The test is blind to crossing
hazard configurations where early and late hazard differences cancel
through integration against the $\gamma$-weighted pooling of the two arms.

%
\section{Detection Coherence Under Domain Restrictions}\label{Sect:restriction}

The four tests analyzed in the previous section can achieve detection
coherence under appropriately restricted domains.

\subsection{General Principle}

A detection-incoherent test over $\mathcal{P}$ is detection coherent over
a restricted domain $\mathcal{R} \subsetneq \mathcal{P}$
if and only if $\mathcal{S}_{\mathcal{P}} \cap \mathcal{R} = \emptyset$,
i.e., no blind spot distribution belongs to~$\mathcal{R}$.

Throughout this section, it is assumed the equivalence of Proposition~\ref{prop:blindequiv} holds.
The theorems below are stated and proved in this equivalent form.

\subsection{Wilcoxon--Mann--Whitney Test Under Monotone Likelihood Ratio Families}

We demonstrate that the WMW test achieves
detection coherence when restricted to monotone likelihood ratio
(MLR) families.

\begin{definition}[Stochastic order]
$F_1 \preceq_{st} F_2$ denotes that $F_2$ stochastically
dominates $F_1$: $F_1(x) \geq F_2(x)$ for all $x$, equivalently
$\mathrm{P}(X_1 > x) \leq \mathrm{P}(X_2 > x)$ for all $x$, for $X_1\sim F_1$,
$X_2\sim F_2$.
\end{definition}

\begin{definition}[MLR family]
A parametric family $\{f_\theta\mathpunct{:} \theta \in \Theta\}$ has monotone
likelihood ratio if $f_{\theta_2}(x)/f_{\theta_1}(x)$ is monotone
in $x$ for all $\theta_1 < \theta_2$.
\end{definition}

\begin{theorem}[WMW coherence under MLR restriction]
\label{thm:wmw-mlr}
The Wilcoxon--Mann--Whitney test is detection coherent when restricted
to domains where $(F_1, F_2)$ form an MLR family. The WMW detection
functional carries no auxiliary argument
(Appendix~\ref{App:WMW}), so no allocation indexing arises.
\end{theorem}

\begin{proof}
Let $\mathcal R_{MLR}$ denote the domain of distribution pairs forming
MLR families. We must show
$\mathcal{H}_0^c \cap \mathcal{R}_{MLR} = \mathcal{D} \cap
\mathcal{R}_{MLR}$, where $\mathcal{H}_0^c = \{(F_1,F_2)\mathpunct{:} F_1 \neq F_2\}$
and $\mathcal{D} = \{(F_1,F_2)\mathpunct{:} \mathrm{AUC} \neq 1/2\}$.

The inclusion $\mathcal{D} \cap \mathcal{R}_{MLR} \subseteq
\mathcal{H}_0^c \cap \mathcal{R}_{MLR}$ is immediate: if
$\mathrm{AUC} \neq 1/2$ then $F_1 \neq F_2$.

For the reverse inclusion, suppose $(F_1, F_2) \in \mathcal{H}_0^c
\cap \mathcal{R}_{MLR}$, so $F_1 \neq F_2$ and the distributions
form an MLR family. By the MLR property, if $F_1 \neq F_2$, then
one distribution stochastically dominates the
other~\cite{lehmann1955ordered}. Stochastic dominance implies
$\mathrm{AUC} \neq 1/2$: if $F_1 \preceq_{st} F_2$, then
$\mathrm{P}(X_1 < X_2) > 1/2$, so $\mathrm{AUC} > 1/2$; if
$F_2 \preceq_{st} F_1$, then $\mathrm{AUC} < 1/2$. Hence
$(F_1,F_2) \in \mathcal{D} \cap \mathcal{R}_{MLR}$.

This establishes $\mathcal{H}_0^c \cap \mathcal{R}_{MLR} =
\mathcal{D} \cap \mathcal{R}_{MLR}$, proving detection coherence.
\end{proof}

\begin{remark}[Location-shift as special case]
Location-shift families, where $F_2(x) = F_1(x - \delta)$, form a
special case of MLR families. When $f_1$ has a monotone likelihood
ratio, the ratio $f_2(x)/f_1(x) = f_1(x - \delta)/f_1(x)$ is
monotone in $x$. Therefore, the WMW test is
detection coherent when restricted to location-shift families.
In this context, testing $H_0\mathpunct{:} F_1 = F_2$ becomes equivalent to
testing $H_0\mathpunct{:} \delta = 0$ (equal locations), and $\mathrm{AUC} = 1/2$
occurs precisely when locations are equal.
\end{remark}

  %
%
\subsection{Logrank Test Under Ordered Hazards Families}

We demonstrate that the logrank test achieves detection coherence
when restricted to ordered hazards families.

\begin{definition}[Ordered hazards family]
A pair of hazard functions $(\lambda_1, \lambda_2)$ belongs to the
ordered hazards family if
$$
  \mathcal{F}_{\text{ordered}} = \{(\lambda_1, \lambda_2) :
      \lambda_1(t) \geq \lambda_2(t) \ \forall\, t \quad \text{or} \quad
    \lambda_1(t) \leq \lambda_2(t) \ \forall\, t\}.
$$
  \end{definition}

\begin{theorem}[Logrank coherence under ordered hazards restriction]
\label{thm:logrank-ordered}
The logrank test is detection coherent when restricted to domains
where $(\lambda_1, \lambda_2) \in \mathcal{F}_{\text{ordered}}$.
Coherence holds at every allocation $\gamma \in (0,1)$, as
Definition~\ref{Def:auxiliary-generalization} requires.
\end{theorem}

\begin{proof}
Let $\mathcal{R}_{\text{ordered}}$ denote the domain of
distribution pairs with ordered hazard functions. Fix
$\gamma \in (0,1)$. We must show
$\mathcal{H}_0^c \cap \mathcal{R}_{\text{ordered}} =
  \mathcal{D}(\gamma) \cap \mathcal{R}_{\text{ordered}}$, where
$\mathcal{H}_0^c = \{(S_1,S_2)\mathpunct{:} S_1(t) \neq S_2(t)
  \text{ for some } t\}$ and
$\mathcal{D}(\gamma) = \{(S_1,S_2)\mathpunct{:} \int_0^\infty
  w(t;\gamma)[\lambda_1(t)-\lambda_2(t)]\,dt \neq 0\}$,
$w(t;\gamma) = \gamma(1-\gamma)y_1(t)y_2(t)/
  \{\gamma y_1(t)+(1-\gamma)y_2(t)\}$.

The inclusion $\mathcal{D}(\gamma) \cap \mathcal{R}_{\text{ordered}}
\subseteq \mathcal{H}_0^c \cap \mathcal{R}_{\text{ordered}}$
is immediate: if $\int_0^\infty w(t;\gamma)[\lambda_1(t)-
\lambda_2(t)]\,dt \neq 0$ then $\lambda_1 \neq \lambda_2$
on a set of positive measure, hence $S_1 \neq S_2$.

For the reverse inclusion, suppose $(S_1,S_2) \in
\mathcal{H}_0^c \cap \mathcal{R}_{\text{ordered}}$, so
$S_1 \neq S_2$ and $(\lambda_1,\lambda_2) \in
\mathcal{F}_{\text{ordered}}$. Under ordered hazards,
$\lambda_1(t) - \lambda_2(t)$ does not change sign.
The weight satisfies $w(t;\gamma) > 0$ for all $t$ in the
support and all $\gamma \in (0,1)$, since $\gamma(1-\gamma) > 0$
and $y_1(t), y_2(t) > 0$ there. Therefore
$w(t;\gamma)[\lambda_1(t)-\lambda_2(t)]$ does not change sign
and is non-zero on a set of positive measure, giving
$\int_0^\infty w(t;\gamma)[\lambda_1(t)-\lambda_2(t)]\,dt
\neq 0$, hence $(S_1,S_2) \in \mathcal{D}(\gamma) \cap
\mathcal{R}_{\text{ordered}}$.

This establishes $\mathcal{H}_0^c \cap
\mathcal{R}_{\text{ordered}} = \mathcal{D}(\gamma) \cap
\mathcal{R}_{\text{ordered}}$. The argument used no property
of $\gamma$ beyond $\gamma(1-\gamma) > 0$, so it holds at every
$\gamma \in (0,1)$, and the restricted test is detection
coherent in the sense of
Definition~\ref{Def:auxiliary-generalization}.
\end{proof}

\begin{remark}[Proportional hazards as special case]
Proportional hazards, where $\lambda_1(t) = \kappa\lambda_2(t)$ for
a constant $\kappa > 0$ (written $\kappa$, not $\gamma$, to avoid
collision with the allocation), form a special case of
$\mathcal{F}_{\text{ordered}}$: if $\kappa \geq 1$ then
$\lambda_1(t) \geq \lambda_2(t)$ for all $t$, and if $\kappa \leq 1$
then $\lambda_1(t) \leq \lambda_2(t)$ for all $t$. The ordered
hazards family is strictly larger: for example, $\lambda_1(t) = 1$
and $\lambda_2(t) = 1 + t$ satisfies $\lambda_1(t) > \lambda_2(t)$
for all $t > 0$ but has a non-constant hazard ratio
$\lambda_1(t)/\lambda_2(t) = 1 + t$, hence does not belong to any
proportional hazards family. Thus
$\mathcal{F}_{\text{PH}} \subsetneq \mathcal{F}_{\text{ordered}}$
strictly, and the logrank test is detection coherent under a
broader restriction than proportional hazards alone.
\end{remark}

\subsection{Other Tests Under Appropriate Restrictions}

For the KW and Friedman tests, we state the known domain
restrictions under which detection coherence is achieved, without
formal proof, as the restrictions are standard in the literature.

\textbf{Kruskal--Wallis test under location-shift:} Under the
restriction $F_k(x) = F(x - \mu_k)$, the vector detector
$(\bar{R}_1/(n+1),\ldots,\bar{R}_K/(n+1))$
reflects relative location effects. Fix an allocation $\gamma$ on the
simplex with $\gamma_k > 0$ for all $k$. If the locations are not all
equal, the $\gamma$-weighted pooled mixture $\sum_r \gamma_r F_r$ is a
fixed distribution against which the extreme groups compare
strictly on opposite sides of $1/2$, so
$(A_1(\gamma),\ldots,A_K(\gamma)) \neq (1/2,\ldots,1/2)$. Hence when
the detection-null value is attained the locations coincide and
$F_1 = \cdots = F_K$. The argument uses only $\gamma_k > 0$, so it
holds at every allocation, as
Definition~\ref{Def:auxiliary-generalization} requires; the
alternative coincides with the detection set and the blind spot is
empty.

\textbf{Friedman test under block-consistent ordering:}
Suppose the within-block treatment distributions are consistently ordered
across all blocks: $F_{i1} \preceq_{st} \cdots \preceq_{st} F_{ik}$
for all $i = 1,\ldots,n$ (or the reverse ordering throughout).
The Friedman test carries no auxiliary argument
(Appendix~\ref{app:friedman-dependence}), so no allocation indexing
arises here.
Under this restriction, if any two treatment distributions differ within any block,
the ordering is strict on a set of positive measure, so $\mu_j \neq \mu_{j'}$
for some $j, j'$, and $(\mu_1,\ldots,\mu_k) \neq \theta^*$.
When the detection-null value is attained, $\mu_j = (k+1)/2$ for all~$j$,
which under consistent ordering forces $F_{i1} = \cdots = F_{ik}$ for all~$i$.
The alternative thus coincides with the detection set, eliminating the blind spot.

%
%
\section{Three Detection-Coherent Nonparametric Tests}\label{Sect:two}
  
The Kolmogorov--Smirnov test~\cite{an1933sulla, smirnov1939},
the Zaremba test~\cite{zaremba1962}, and the MMD test with a characteristic kernel~\cite{gretton2012},
are universally detection coherent, as we now demonstrate.
For the KS test, coherence follows directly from Definition~\ref{Def:detection}.
For Zaremba's test, coherence is established via the assessment
procedure of Section~\ref{Sect:recipe}.
For the MMD test with a characteristic kernel, coherence is established via the
assessment procedure and with the use of results from~\cite{gretton2012}.

For the KS test, the null hypothesis is
$\mathcal{H}_0 = \{(F, G)\mathpunct{:} F = G\}$,
the detector is $\sup_x |\hat{F}_n(x) - \hat{G}_n(x)|$,
the detection functional is $\theta(F, G) = \sup_x |F(x) - G(x)|$,
the detection-null value is $\theta^* = 0$, and the detection-null set is
$\mathcal{D}_0 = \{(F, G)\mathpunct{:} \sup_x |F(x) - G(x)| = 0\} =
\{(F, G)\mathpunct{:} F = G\}$, since $\sup_x |F(x) - G(x)| = 0$ if and only if
$F = G$. Since $\mathcal{D}_0 = \mathcal{H}_0$,
the blind spot $\mathcal{S}_{\mathcal{P}} = \mathcal{D}_0 \setminus \mathcal{H}_0$
is empty and the test is detection coherent.

For the Zaremba test, the null hypothesis is
$H_0\mathpunct{:} \mathrm{AUC} = 1/2$, the calibration statistic is
$Z_{\mathrm{Z}} = \sqrt{n}(\mathrm{eAUC} - 1/2)/\hat{\sigma}$, the detector
is $\mathrm{eAUC} = U/(n_1 n_2)$, the detection functional
is $\theta(F_1, F_2) = \mathrm{AUC} = \mathrm{P}(X < Y) + \tfrac{1}{2}\mathrm{P}(X = Y)$,
the detection-null value is $\theta^* = 1/2$, and the detection-null
set is $\mathcal{D}_0 = \{(F_1, F_2)\mathpunct{:} \mathrm{AUC} = 1/2\}$. Since the
null hypothesis coincides with the detection-null set,
the blind spot $\mathcal{S}_{\mathcal{P}} =
\mathcal{D}_0 \setminus \mathcal{H}_0$ is empty and the test is
detection coherent.

For the MMD test with a characteristic kernel, the null hypothesis is
$\mathcal{H}_0 = \{(F, G)\mathpunct{:} F = G\}$, the detector is the unbiased
$U$-statistic $\widehat{\text{MMD}}_u^2$~\cite[Lemma~6, Eq.~3]{gretton2012},
consistent for $\text{MMD}^2(F,G)$ on all $\mathcal P$~(\cite[Thm.~7]{gretton2012},
stated for the biased statistic, which differs from the unbiased one by $O(1/n)$);
the detection functional is $\theta(F,G) = \text{MMD}^2(F,G)$, and the detection-null set is
$\mathcal{D}_0 = \{(F,G)\mathpunct{:} \text{MMD}^2(F,G) = 0\}$. When the kernel is characteristic,
$\text{MMD}^2(F,G) = 0$ if and only if $F = G$ (\cite[Thm. 5]{gretton2012} and discussion following),
so $\mathcal{D}_0 = \mathcal{H}_0$ and the test is detection coherent.

\begin{note}[The three tests carry no auxiliary argument]
\label{Note:two-aux}
None of the three detection functionals above depends on an argument
besides the pair of distributions, so $\Xi$ is a singleton in each case
(Definition~\ref{Def:auxiliary}) and
Definitions~\ref{Def:detection}--\ref{Def:calcoh} apply in their
unindexed form.

For the Zaremba test the detector is $\mathrm{eAUC} = U/(n_1 n_2)$, the
same statistic as for the WMW test, and the argument of
Appendix~\ref{App:WMW} applies verbatim.

For the KS test each marginal is estimated from its own
sample, $\hat{F}_{n_1}$ from the first and $\hat{G}_{n_2}$ from the
second, with no pooling. By the Glivenko--Cantelli theorem
$\sup_x|\hat{F}_{n_1}(x) - F(x)| \to 0$ and
$\sup_x|\hat{G}_{n_2}(x) - G(x)| \to 0$ almost surely, while the
triangle inequality gives
$$
\Bigl|\, \sup_x|\hat{F}_{n_1}(x) - \hat{G}_{n_2}(x)|
       - \sup_x|F(x) - G(x)| \,\Bigr|
\;\leq\;
\sup_x|\hat{F}_{n_1}(x) - F(x)| + \sup_x|\hat{G}_{n_2}(x) - G(x)| .
$$
The detector therefore converges to $\theta(F,G) = \sup_x|F(x)-G(x)|$
along any sequence with $n_1, n_2 \to \infty$, whatever the ratio
$n_1/n$.

For the MMD test the unbiased statistic is a sum of three terms, each
normalized by its own sample sizes~\cite[Eq.~3]{gretton2012}:
$$
\widehat{\mathrm{MMD}}^2_u =
\frac{1}{n_1(n_1-1)}\sum_{i \neq j} k(X_i, X_j)
+ \frac{1}{n_2(n_2-1)}\sum_{i \neq j} k(Y_i, Y_j)
- \frac{2}{n_1 n_2}\sum_{i=1}^{n_1}\sum_{j=1}^{n_2} k(X_i, Y_j).
$$
The first two are one-sample $U$-statistics and the third a two-sample
$U$-statistic; under $\mathrm{E}|k| < \infty$ for each pair, automatic
for a bounded kernel, the law of large numbers for $U$-statistics gives
$\mathrm{E}\,k(X,X')$, $\mathrm{E}\,k(Y,Y')$ and $\mathrm{E}\,k(X,Y)$
respectively, as $n_1 \to \infty$ and $n_2 \to \infty$ separately.
Hence $\theta(F,G) = \mathrm{MMD}^2(F,G)$, again free of the
allocation.
\end{note}

Since detection coherence with respect to $\mathcal{P}$ is universal
detection coherence by definition, the three tests are universally
detection coherent.

\begin{note}[Coherence is necessary, not sufficient]
\label{Note:necessary}
Detection coherence concerns only the relation between
$\mathcal{H}_0$ and $\mathcal{D}_0$. It says nothing about whether
a test holds its nominal level throughout $\mathcal{H}_0$, about
its power at any fixed sample size, or about the accuracy of the
reference distribution in small samples. An empty blind spot
is a necessary condition for consistency against every violation of
the null (Corollary~\ref{cor:calcohequiv} and
Proposition~\ref{prop:power}), not a certificate that the test is
otherwise well behaved.
\end{note}

%
\section{Discussion}\label{Sect:discussion}

\subsection{Nonparametric testing framework}
  
In the nonparametric testing framework (cf.~\cite{hettmansperger1984,lehmann2022}),
a test statistic is calibrated under a null $H_0$, an alternative
of interest $H_A$ is specified -- typically on interpretability or convenience
grounds, not because it exhausts the space of departures from
$H_0$ -- and consistency of the test against $H_A$ is demonstrated.
The region $(H_0 \cup H_A)^c \subset \mathcal{P}$ is left uncharacterised.
Data, however, arrive from  $\mathcal{P}$, not from $H_0\cup H_A$.
Gnedenko's 1958 example for the WMW test (see~\S\ref{sec:history-wmw}) is a distribution
pair in exactly this uncharacterised region.

The detection coherence framework
addresses this gap. The detection-null
$\mathcal{D}_0$, the set of
distributions at which the detection
functional takes its null value, is
determined by the test statistic, not
chosen by the analyst. Detection
coherence requires $\mathcal{H}_0 =
\mathcal{D}_0$ over $\mathcal{P}$;
the null and the detection-null
coincide over the full distribution
space. By Corollary~\ref{cor:calcohequiv},
this is equivalent to $\mathcal{H}_0^c =
\mathcal{D}$, so that $\mathcal{H}_0 \cup
\mathcal{D} = \mathcal{P}$ and no
region of $\mathcal{P}$ is left
uncharacterised. A test that fails
this requirement is detection
incoherent, as an unrestricted test.

A foundational requirement of test construction
(cf. Cox~\cite[pp.~31--32]{cox2006principles})
is that large values of the test statistic indicate departure from ${H}_0$.
The detection coherence requirement of an empty blind spot guarantees
that every two-sided violation of ${H}_0$ is detectable
(Corollary~\ref{cor:calcohequiv}).
Together, the foundational requirement and detection coherence ensure
consistency of the test against every violation of ${H}_0$: large values of the test statistic
indicate departure from ${H}_0$, and for any departure from ${H}_0$, power grows
to~$1$ with sample size.

\subsection{Two partitions}

A test defined over a domain
$\mathcal{R}$ induces two partitions
of $\mathcal{R}$. The hypothesis
partition $\{\mathcal{H}_0,
\mathcal{H}_0^c\}$ is determined by
the null hypothesis under which the
test is calibrated. The detection
partition $\{\mathcal{D}_0, \mathcal{D}\}$
is determined by the calibration
statistic; $\mathcal{D}_0$ is the set
of distributions in $\mathcal{R}$ at
which the detection functional takes
its null value $\theta^*$, and
$\mathcal{D} = \mathcal{R} \setminus
\mathcal{D}_0$ is its complement.

Detection coherence requires
$\mathcal{H}_0 = \mathcal{D}_0$ over
$\mathcal{R}$: the null and the
detection-null coincide. When they
do not -- $\mathcal{H}_0 \subsetneq
\mathcal{D}_0$ -- the excess
$\mathcal{D}_0 \setminus \mathcal{H}_0$
is the blind spot.

\subsection{The incoherent tests}

The four tests
examined (Wilcoxon--Mann--Whitney, Kruskal--Wallis,
Friedman, and logrank) are detection incoherent as
unrestricted tests (Propositions~\ref{prop:wmw-incoherent},
\ref{prop:kw-incoherent}, \ref{prop:friedman-incoherent},
\ref{prop:logrank-incoherent}): their null hypotheses
$\mathcal{H}_0$ are strictly smaller than their
detection-null sets $\mathcal{D}_0$, creating
non-empty blind spots.

\subsection{Consequences of detection incoherence}\label{Subsect:consequences}

\subsubsection*{When non-rejecting}
The blind spot distorts inference when an incoherent test does not reject.
The two-sided alternatives of a detection incoherent test partition
into the detection set and the blind spot. A non-significant result from such a test
therefore has three possible sources:
\textit{i)}~a true negative from $\mathcal{H}_0$,
\textit{ii)} a false negative against a detection
set alternative $\mathcal{D}$, or
\textit{iii)} a false negative against a blind spot $\mathcal{S}_{\mathcal P}$
alternative -- \emph{missed discovery}.

As $n \to \infty$, false negatives against the detection set vanish;
those against the blind spot -- missed discoveries -- persist.

\subsubsection*{When rejecting}
The blind spot distorts inference when an incoherent test rejects. A
significant result from such a test has three possible sources:
\textit{i)} a true positive from $\mathcal{D}$,
\textit{ii)} a false positive against $\mathcal{H}_0$, or
\textit{iii)} a true positive from the blind spot $\mathcal{S}_{\mathcal P}$
-- \emph{spurious discovery}.

As $n \to \infty$, the Type~I error rate stays at $\alpha$  and the spurious discovery
persists at its own fixed rate $c(P^*, \gamma)$. The dual reading is what makes it 
spurious: relative to $\mathcal{H}_0$ the rejection is a true positive, the null being false; 
relative to $\mathcal{D}_0$ it is not, occurring at a rate fixed by the distribution 
rather than in response to anything the test detected.

\subsubsection*{Summary}

Table~\ref{tab:inference} summarizes inferential consequences of detection
incoherence.

For a coherent test, $\mathcal{H}_0^c = \mathcal{D}$, and the classical
$2\times2$ confusion matrix applies directly: a true negative or false
positive at rate $1-\alpha$ or $\alpha$ when $P\in\mathcal{H}_0$; a false
negative or true positive at rate $\beta(P)$ or $1-\beta(P)$ when
$P\in\mathcal{D}$, with $\beta(P)\to0$ as $n\to\infty$ for any fixed $P$
(consistency).

For an incoherent test, $\mathcal{H}_0^c = \mathcal{D}\cup\mathcal{S}_{\mathcal
P}$, a disjoint union, and the alternative side of the confusion matrix
splits into two behaviorally distinct cases. The null side is unaffected
by incoherence.

\renewcommand{\arraystretch}{1.15}
\begin{table}[h]
\setlength{\belowcaptionskip}{-5pt}
\smallskip
\begin{center}
\begin{tabular}{lllll}
\hline
True state & Outcome & Rate & As $n\to\infty$ & \\
\hline
$P\in\mathcal{H}_0$ & true negative & $1-\alpha$ & fixed at $1-\alpha$ & \\
$P\in\mathcal{H}_0$ & false positive & $\alpha$ & fixed at $\alpha$ & \\
$P\in\mathcal{D}$ & false negative & $\beta(P)$ & $\to 0$ & \\
$P\in\mathcal{D}$ & true positive & $1-\beta(P)$ & $\to 1$ & \\
$P^*\in\mathcal{S}_{\mathcal P}$ & false negative & $1-c(P^*,\gamma)$ & fixed -- \emph{missed discovery} & \\
$P^*\in\mathcal{S}_{\mathcal P}$ & true positive & $c(P^*,\gamma)$ & fixed -- \emph{spurious discovery} & \\
\hline
\end{tabular}
\end{center}
\caption{Inferential consequences of detection incoherence.}
\label{tab:inference}
\end{table}
\renewcommand{\arraystretch}{1}

The missed discovery and the spurious discovery are
permanent: no increase in sample size recovers the undetected departure,
nor changes the rate of spurious discoveries.

Thus, using WMW to test $F_1 = F_2$, KW to test $F_1 = \cdots = F_K$,
Friedman to test $F_{i1} = \cdots = F_{ik}$, and logrank to test $S_1(t) = S_2(t)$
should be abandoned as unrestricted tests, due to their non-empty blind spots.
Restricting their domain can in principle eliminate the blind spot, but,
as the next section shows, this is not a reliable remedy in practice.

\subsection{Restriction}

A domain restriction changes the
reference space from $\mathcal{P}$ to
$\mathcal{R} \subsetneq \mathcal{P}$;
within the model $\mathcal{R}$ the test may be
coherent: $\mathcal{H}_0 =
\mathcal{D}_0$ over $\mathcal{R}$.
Using any test
under a model restriction introduces
well-known difficulties: the model
must be verified from the data, any
such verification is uncertain at
finite sample sizes, and the combined
verify-then-test procedure has
uncontrolled size. These apply
equally to coherent and incoherent
tests and are not further discussed
here.

For detection incoherent tests,
restriction introduces two additional
problems specific to incoherence.

First, the named model --
monotone likelihood ratio for WMW,
ordered hazards for logrank,
location-shift for KW,
block-consistent stochastic ordering for Friedman
-- is a sufficient condition for
$\mathcal{H}_0 = \mathcal{D}_0$ over
$\mathcal{R}$, not a characterisation
of the maximal coherence-preserving
domain $\mathcal{R}^*$, defined as
the largest domain over which
$\mathcal{H}_0 = \mathcal{D}_0$,
i.e.\ $\mathcal{R}^* = \{P\mathpunct{:}
\theta(P)\!=\!\theta^*\!\Rightarrow\!
P\!\in\!\mathcal{H}_0\}$.
Outside
the named model but inside
$\mathcal{R}^*$, the test remains
coherent; outside $\mathcal{R}^*$,
the blind spot is non-empty and
$\mathcal{H}_0 \subsetneq \mathcal{D}_0$
over $\mathcal{P}$. The distance
from the named model to the boundary
of $\mathcal{R}^*$ is unknown.

Second, near the boundary of
$\mathcal{R}^*$, $\theta(P) \approx
\theta^*$: the detection functional
takes a value close to its null
value despite $P \notin \mathcal{H}_0$.
Power against such distributions still tends to~$1$, 
but at a rate governed by the magnitude of $\theta(P)-\theta^*$, 
which vanishes as $P$ approaches the boundary: 
for every sample size there are alternatives in $\mathcal{R}^*$ 
against which the test's power is close to its level, 
and their distance from $\mathcal{H}_0$ does not shrink as that happens.
A non-significant result is therefore compatible both with a distribution 
close to $\mathcal{H}_0$ and with one remote from it, 
and does not discriminate between them.

These two considerations compound
rather than cancel. The first says
the analyst cannot locate the
boundary of $\mathcal{R}^*$; the
second says its vicinity cannot be
recognized from the test's own
behavior either.
Neither the model
assumption nor the observed outcome
supplies the reassurance the other
lacks. 
Neither consideration depends on which sufficient condition is named. 
Any domain over which the test is coherent is either a proper subset of $\mathcal{R}^*$, 
in which case both considerations apply unchanged, or $\mathcal{R}^*$ itself -- and restricting 
to $\mathcal{R}^*$ assumes what one wanted to establish, that the test is not 
blind where it is used.
A test should
be used over a domain where
$\mathcal{H}_0 = \mathcal{D}_0$
without requiring any model
restriction. Detection incoherent
tests should be abandoned, not
salvaged through restrictions.

\subsection{Scope}

Every nonparametric test in current use should be
assessed for detection coherence; the seven tests
examined here include some of the most widely used, but
they are not the only candidates.

The present framework provides the theoretical
foundation that explains why calibration under the
detection-null set is not merely a technical
refinement but a necessary condition for coherent
inference. It is thus relevant also for the ongoing
broader programme of rank-based inference for
factorial designs and general nonparametric settings
(cf.~\cite{brunner2017rank, brunner2018rank,
chung2013exact}, among others).

%
\section{Historical Note}
\label{sec:history}

The problem the framework addresses was first
identified, for the WMW test specifically, by~Gnedenko~\cite{gnedenko1958}.

\subsection{Gnedenko's recognition of the WMW test flaw}
\label{sec:history-wmw}

Gnedenko~\cite{gnedenko1958} constructed two distributions
with $F_1\!\neq\!F_2$ and $\mathrm{AUC}\!=\!1/2$,
and demonstrated that the WMW test erroneously accepts ${H}_0\mathpunct{:} F_1\!=\!F_2$
with probability tending to~$1$. He closed with a practical warning
(\emph{translated from Russian}): ``The example given above
shows that if, in practical applications,
the Wilcoxon criterion gives a result of coincidence of distributions,
then this conclusion should be treated with doubt.''

Zaremba~\cite{zaremba1962} analysed a simplified form of
Gnedenko's example: $\mathrm{P}[X\!=\!0]=1$ and
$\mathrm{P}[Y\!=\!r]=\mathrm{P}[Y\!=\!s]\!=\!1/2$ with $rs\!<\!0$.
Direct calculation gives $\mathrm{E}\,U\!=\!n_1 n_2/2$ ($\mathrm{AUC}\!=\!1/2$,
blind spot confirmed) and $\mathrm{Var}(U)\!=\!n_1^2 n_2/4$
(true variance at the blind spot distribution).
The ratio $d^2$ of the true to tie-corrected null variance approaches
$16/9$ for equal sample sizes, giving power $c(P^*,\gamma)\!\approx\!0.14$,
non-growing with sample size, the defining characteristic of the blind spot.
Gnedenko's conclusion that power converges to $0$ corresponds to the
regime $n_1$ fixed, $n_2\!\to\!\infty$, in which $d^2\!\to\!0$.

\subsection{The Coherent Reformulation: Zaremba}
\label{sec:history_zaremba}

Zaremba~\cite{zaremba1962}, after referring to Gnedenko's example,
calibrated the WMW statistic (eAUC) under
${H}_0^*\mathpunct{:} \mathrm{AUC}\!=\!1/2$ rather than ${H}_0\mathpunct{:} F_1\!=\!F_2$,
deriving the variance of the WMW statistic consistently under ${H}_0^*$
without assuming $F_1\!=\!F_2$.

\subsection{The Behrens-Fisher problem: Fligner and Policello, Brunner and Munzel}
\label{sec:history_bm}

Fligner and Policello~\cite{fligner1981}, under assumption of symmetry of distributions,
and, later,  Brunner and
Munzel~\cite{brunner2000nonparametric}, in full generality, each arrived at
a coherent test of $H_0^*\colon
\mathrm{AUC} = 1/2$ by addressing
the Behrens--Fisher problem.

%
%
\appendix

\section{Verification of Detector Properties}
\label{App:detector}

For each of the four tests examined in
Section~\ref{Sect:major}, we apply the extraction procedure of
Definition~\ref{Def:extraction} and verify that the extracted
quantity satisfies the consistency and minimality conditions of
Definition~\ref{Def:detection}.
The extraction procedure is a backward engineering
operation: given the calibration statistic~$T_n$,
which converges in distribution under $H_0$, it
recovers the underlying detector $\hat{D}_n$ whose
convergence in probability to the detection functional
$\theta(P)$ drives that convergence in distribution.
For ratio statistics (WMW, logrank) this backward
engineering is mechanical; for quadratic form statistics
(KW, Friedman) it requires identifying the minimal
vector pre-image, guided by the consistency check.
For each test, we also assess
whether the resulting detection functional depends on an
auxiliary argument (Definition~\ref{Def:auxiliary}).

Throughout this appendix, we work with representations
of the calibration statistic that make the sample size
scaling explicit, exploiting the fact that the calibration
statistic is defined only up to monotone rescaling
(Remark~\ref{Rem:monotone}).

\subsection{General Minimality Lemma}

\begin{lemma}[Minimality via non-degenerate range]
\label{Lem:minimality}
Let $\hat{D}_n\!\xrightarrow{p}\!\theta(P)$ and let
$\Theta\!=\!\{\theta(P)\mathpunct{:} P\!\in\!\mathcal{R}\}$ denote the
attainable range of the detection functional. 
Then no non-injective continuous
function $f\mathpunct{:} \Theta\!\to\!\mathbb{R}^d$ satisfies
$f(\hat{D}_n)\!\xrightarrow{p}\!\theta(P)$ for all $P$.
\end{lemma}

\begin{proof}
Since $f$ is non-injective, there exist $a \neq b$ in
$\Theta$ with $f(a) = f(b)$. Since $\theta(P)$ attains
all values in $\Theta$, there exist $P_1, P_2 \in
\mathcal{P}$ with $\theta(P_1) = a$ and $\theta(P_2) = b$.
Under $P_1$: $f(\hat{D}_n) \xrightarrow{p} f(a)$. Under
$P_2$: $f(\hat{D}_n) \xrightarrow{p} f(b) = f(a)$. Since
$\theta(P_1) = a \neq b = \theta(P_2)$, $f(\hat{D}_n)$
cannot consistently estimate $\theta(P)$ for both $P_1$
and $P_2$ simultaneously. Contradiction.
\end{proof}

\subsection{Wilcoxon--Mann--Whitney test}\label{App:WMW}

{Conventional calibration statistic:}
$Z_{\mathrm{WMW}} = \sqrt{n}(\mathrm{eAUC} - 1/2)/\hat{\sigma}$,
where $\mathrm{eAUC} = U/(n_1 n_2)$ and
$U = \sum_{i=1}^{n_1}\sum_{j=1}^{n_2}\Bigl[{1}\{X_i < Y_j\} + \tfrac12{1}\{X_i = Y_j\}\Bigr]$,
asymptotically $N(0,1)$ under $H_0\mathpunct{:} F_1 = F_2$.

{Convenience version:} Same (the scaling $\sqrt{n}$
is already explicit). $Z_{\mathrm{WMW}}$ is a ratio statistic:
$$
Z_{\mathrm{WMW}} = \sqrt{n} \cdot \frac{\mathrm{eAUC} - 1/2}
{\hat{\sigma}}
$$
  
{Step 1 (Remove scaling):} Remove $a_n = \sqrt{n}$:
$$\frac{\mathrm{eAUC} - 1/2}{\hat{\sigma}}$$

{Step 2 (Remove variance normalization):} Remove
$\hat{\sigma}$, consistently estimating $\sigma(P)$, the asymptotic
standard deviation of $\sqrt{n}(\mathrm{eAUC} - \mathrm{AUC})$:
$$\mathrm{eAUC} - 1/2$$
  
{Step 3 (Remove nuisance):}
The additive constant $1/2$ is absorbed into the detection-null value
and is removed, leaving $\mathrm{eAUC}$.

{Step 4 (Minimal consistent estimator):}
The remaining quantity is $\mathrm{eAUC}$.

{Consistency check:} By $U$-statistic
theory~\cite{hoeffding1948class}:
$$\mathrm{eAUC} \xrightarrow{p} \mathrm{AUC} = \mathrm{P}(X < Y) + \frac{1}{2}\mathrm{P}(X = Y)$$
Consistency passes.

{Detector:} $\hat{D}_n = \mathrm{eAUC}$,
consistently estimating the detection functional
$\theta(F_1,F_2) = \mathrm{AUC}$.

{Minimality:} Immediate from
Lemma~\ref{Lem:minimality}; the attainable range is
$\Theta = [0,1]$.

{Recovery of calibration statistic:}
$$Z_{\mathrm{WMW}} = \frac{\sqrt{n}(\mathrm{eAUC} - 1/2)}
{\hat{\sigma}} + o_p(1)$$
where $\hat{\sigma}^2$ consistently estimates the
asymptotic variance of $\sqrt{n}(\mathrm{eAUC} -
\mathrm{AUC})$ by $U$-statistic
theory~\cite{hoeffding1948class}.

\subsubsection*{Auxiliary-argument dependence.}
The detector is $\mathrm{eAUC} = U/(n_1n_2)$, with
$\mathrm{eAUC} \xrightarrow{p} \mathrm{AUC}$.
The limit is a functional of $(F_1,F_2)$ alone. 
The group sizes enter the detector only through the normalization $n_1n_2$, 
and the limit is the same along any sequence with $n_1, n_2 \to \infty$, 
whatever the ratio $n_1/n$. No allocation enters $\theta$, $\Xi$ is a singleton 
and Definitions~\ref{Def:detection}--\ref{Def:detectionset} apply in their 
unindexed form.

\subsection{Kruskal--Wallis test}\label{App:KW}

{Conventional calibration statistic:}
$H = \frac{12}{n(n+1)}\sum_{k=1}^K n_k
(\bar{R}_k - (n+1)/2)^2$,
asymptotically $\chi^2_{K-1}$ under $H_0$.

{Structure:} Unlike WMW and logrank, $H$ is a
quadratic form (not a ratio statistic). The extraction
procedure identifies the minimal vector pre-image of
the quadratic form, guided by the consistency check.

{Convenience version:} $H/n$, with explicit
sample size scaling $n$.

{Step 1 (Remove scaling):} Remove $n$:
$$\frac{H}{n} = \frac{12}{n+1}\sum_{k=1}^K \frac{n_k}{n}
\left(\bar{R}_k - \frac{n+1}{2}\right)^2$$
  
{Step 2 (Remove variance normalization):} Remove
$12/(n+1)$, and rewrite to make the rank normalization
explicit:
$$
  \sum_{k=1}^K \frac{n_k}{n}
  \left(\bar{R}_k - \frac{n+1}{2}\right)^2 =
  (n+1)^2\sum_{k=1}^K \frac{n_k}{n}
  \left(\frac{\bar{R}_k}{n+1} - \frac{1}{2}\right)^2
$$
  
{Step 3 (Remove nuisance and identify vector
pre-image):} Three nuisance components are identified
and removed:
\begin{enumerate}[(i)]
\item $(n+1)^2$, scaling term, remove
\item $n_k/n$, weights with plim $\gamma_k =
  \lim n_k/n$, not depending on $P$, remove
\item $-1/2$, a constant, not depending on $P$, remove
\end{enumerate}
The quadratic form collapses the vector
$(\bar{R}_1/(n+1),\ldots,\bar{R}_K/(n+1))$ to a scalar
non-injectively. The vector is the minimal pre-image.

{Step 4 (Minimal consistent estimator):}
The candidate detector is
$(\bar{R}_1/(n+1),\ldots,\bar{R}_K/(n+1))$.

{Consistency check:} By the rank identity of
Section~\ref{Sect:major},
$$
\frac{\bar{R}_k}{n+1}
  = \sum_{r=1}^{K} \hat\gamma_r\,\hat{A}_{rk} + O(n^{-1}),
\qquad \hat\gamma_r = n_r/n, \quad \hat{A}_{kk} = 1/2 ,
$$
with the remainder is bounded by $1/(n+1)$, 
since $0\le\hat A_{rk}\le 1$.
For $r \neq k$, $\hat{A}_{rk}$ is a
two-sample $U$-statistic with bounded kernel, so
$\hat{A}_{rk} \xrightarrow{p} A_{rk}$ as $n_r, n_k \to \infty$ 
(by the $U$-statistic Law of Large Numbers of \cite{lehmann1975nonparametrics});
and $\hat\gamma_r \to \gamma_r$ by assumption. Hence
$$
\frac{\bar{R}_k}{n+1} \xrightarrow{p}
  \sum_{r=1}^{K} \gamma_r A_{rk} = A_k(\gamma)
  = \mathrm{P}(X_\gamma < X_k) + \tfrac12\,\mathrm{P}(X_\gamma = X_k),
$$
the AUC of the pooled mixture against group~$k$. Consistency
passes. Note the orientation: a large mean rank corresponds to large
observations in group~$k$, so the limit is the probability that a
pooled draw falls \emph{below} a draw from group~$k$.
 
{Detector:} $\hat{D}_n =
(\bar{R}_1/(n+1),\ldots,\bar{R}_K/(n+1))$, vector-valued,
consistently estimating the detection functional
$\theta(F_1,\ldots,F_K;\gamma) = (A_1(\gamma),\ldots,A_K(\gamma))$.

{Minimality:} Immediate from Lemma~\ref{Lem:minimality}.

{Detection-null set coincidence:}
Let $\hat{\varepsilon}^2 = H/(n-1)$ denote the
epsilon-squared statistic. Then:
$$\hat{\varepsilon}^2 = 0 \iff \bar{R}_k = (n+1)/2
\ \forall k \iff \frac{\bar{R}_k}{n+1} = \frac{1}{2}
\ \forall k \iff \hat{D}_n = \left(\frac{1}{2},\ldots,
\frac{1}{2}\right) = \theta^*
$$
The detection-null sets of $\hat{\varepsilon}^2$ and
the vector detector coincide, so the blind spot
characterization is unaffected. 
Moreover, $\hat{\varepsilon}^2$ is not a minimal detector since
it is a non-injective function of $\hat{D}_n$.

{Recovery of calibration statistic:}
$$
H = \Phi(\hat{D}_n) = \frac{12(n+1)}{n}\sum_{k=1}^K n_k
\left(\hat{D}_{nk} - \frac{1}{2}\right)^2
$$
where $\Phi$ is the quadratic form mapping the vector
detector to the scalar calibration statistic $H$.

\subsubsection*{Auxiliary-argument dependence.}
\label{app:kw-dependence}
Write
$$
A_{rk} = \mathrm{P}(X_r < X_k) + \tfrac12\,\mathrm{P}(X_r = X_k),
\qquad X_g \sim F_g \ \text{independent},
$$
for the pairwise AUC of group~$r$ against group~$k$, so
that $A_{rk} + A_{kr} = 1$ and $A_{kk} = 1/2$. Since the pooled
reference $X_\gamma$ has distribution $\sum_r \gamma_r F_r$,
$A_k(\gamma)$ is the AUC of the pooled mixture against group~$k$,
and conditioning on which component the draw came from gives 
\begin{equation}
\label{eq:kw-mixture}
A_k(\gamma) = \sum_{r=1}^K \gamma_r\, A_{rk}
            = \tfrac12\gamma_k + \sum_{r \neq k} \gamma_r\, A_{rk}.
\end{equation}

For $K \geq 3$ the weights $\gamma$ range over a simplex of
dimension $K-1 \geq 2$, so the coefficients attached to the cross
terms in~\eqref{eq:kw-mixture} can be varied against one another
while the $A_{kr}$ stay fixed. Whenever the $A_{kr}$ are not all
equal to $1/2$, the vector $(A_1(\gamma),\ldots, A_K(\gamma))$ moves
with $\gamma$, and so does the detection-null set
$\mathcal{D}_0(\gamma) = \{(F_1,\ldots,F_K)\mathpunct{:}
A_k(\gamma) = 1/2 \ \forall k\}$. The detection functional therefore
carries the allocation as an auxiliary argument in the sense of
Definition~\ref{Def:auxiliary}, and $\Xi$ is the simplex, not a
singleton.

At $K = 2$ the dependence is present in the components but cancels
in the comparison that matters. Writing $\gamma = \gamma_1$ and
$\mathrm{AUC} = A_{12}$, \eqref{eq:kw-mixture} gives
$$
A_1(\gamma) = \gamma_1 A_{11} + \gamma_2 A_{21} = \tfrac12\gamma + (1-\gamma)\,(1 - \mathrm{AUC}),
\qquad
A_2(\gamma) = \gamma_1 A_{12} + \gamma_2 A_{22} =   \gamma\,\mathrm{AUC} + \tfrac12(1-\gamma),
$$
so that
$$
A_2(\gamma) - A_1(\gamma)
  = \gamma\bigl(\mathrm{AUC} - \tfrac12\bigr)
  + (1-\gamma)\bigl(\mathrm{AUC} - \tfrac12\bigr)
  = \mathrm{AUC} - \tfrac12,
$$
free of $\gamma$. Both components equal $1/2$ precisely when
$\mathrm{AUC} = 1/2$, whatever the allocation, so
$\mathcal{D}_0(\gamma)$ does not move with $\gamma$ and the
two-sample case reduces to the WMW test, whose detection functional
carries no auxiliary argument
(Appendix~\ref{App:WMW}). This is why the dependence is
stated for $K \geq 3$.

\subsection{Friedman test}\label{App:Friedman}

{Conventional calibration statistic:}
$Q = \frac{12n}{k(k+1)}\sum_{j=1}^k
(\bar{R}_{\cdot j} - (k+1)/2)^2$,
asymptotically $\chi^2_{k-1}$ under $H_0$.

{Structure:} Like KW, $Q$ is a quadratic form.
Unlike KW, within-block ranks $\bar{R}_{\cdot j} \in
[1,k]$ are bounded (of order $O_p(1)$ rather than
$O_p(n)$), making the extraction cleaner.

{Convenience version:} $Q/n$, with explicit
sample size scaling $n$.

{Step 1 (Remove scaling):} Remove $n$:
$$
\frac{Q}{n} = \frac{12}{k(k+1)}\sum_{j=1}^k
\left(\bar{R}_{\cdot j} - \frac{k+1}{2}\right)^2
$$

{Step 2 (Remove variance normalization):} Remove $12/k(k+1)$ (inverse of within-block rank variance):
$$\sum_{j=1}^k\left(\bar{R}_{\cdot j} - \frac{k+1}{2}\right)^2
$$

{Step 3 (Remove nuisance and identify vector pre-image):} Two nuisance components are identified
and removed:
\begin{enumerate}[(i)]
\item $(k+1)/2$, a constant, not depending on $P$, remove
\item The quadratic form (non-injective), collapses
vector to scalar; the vector
$(\bar{R}_{\cdot 1},\ldots,\bar{R}_{\cdot k})$
is the minimal pre-image
\end{enumerate}
Leaving $(\bar{R}_{\cdot 1},\ldots,\bar{R}_{\cdot k})$.

{Step 4 (Minimal consistent estimator):}
The candidate detector is
$(\bar{R}_{\cdot 1},\ldots,\bar{R}_{\cdot k})$.

{Consistency check:} By the within-block rank identity of
Section~\ref{Sect:major},
$$
\bar{R}_{\cdot j} = 1 + \sum_{l \neq j} \hat{A}_{lj},
\qquad
\hat{A}_{lj} = \frac1n \sum_{i=1}^n
  \Bigl[{1}\{X_{il} < X_{ij}\}
      + \tfrac12{1}\{X_{il} = X_{ij}\}\Bigr],
$$
exactly, with no remainder. Each $\hat{A}_{lj}$ is an average of
independent bounded summands across blocks, so by the law of large
numbers
$\hat{A}_{lj} \xrightarrow{p} \bar{A}_{lj}$, and therefore
$$
\bar{R}_{\cdot j} \xrightarrow{p}
  \mu_j = 1 + \sum_{l \neq j} \bar{A}_{lj} \in [1,k].
$$
Consistency passes. Blocks need not be identically distributed: the
summands are independent across~$i$ and bounded, so the limit is the
block-averaged AUC $\bar{A}_{lj}$ rather than a common per-block
value.
 
{Detector:} $\hat{D}_n =
(\bar{R}_{\cdot 1},\ldots,\bar{R}_{\cdot k})$, vector-valued,
consistently estimating $\theta(P) = (\mu_1,\ldots,\mu_k)$.

{Minimality:} Immediate from Lemma~\ref{Lem:minimality}.

{Detection-null set coincidence:}
Let $\hat{W} = Q/(n(k-1))$ denote Kendall's $W$
statistic. Then:
$$\hat{W} = 0 \iff \bar{R}_{\cdot j} = (k+1)/2
\ \forall j \iff \mu_j = (k+1)/2 \ \forall j
\iff \hat{D}_n = \left(\frac{k+1}{2},\ldots,
\frac{k+1}{2}\right) = \theta^*$$
The detection-null sets of $\hat{W}$ and the vector
detector coincide, so the blind spot characterization
is unaffected. Moreover, $\hat{W}$ is not a minimal
detector since it is a non-injective function of
$\hat{D}_n$.

{Recovery of calibration statistic:}
$$Q = \Phi(\hat{D}_n) = \frac{12n}{k(k+1)}\sum_{j=1}^k
\left(\hat{D}_{nj} - \frac{k+1}{2}\right)^2$$
where $\Phi$ is the quadratic form mapping the vector
detector to the scalar calibration statistic $Q$.

\subsubsection*{Auxiliary-argument dependence.}
\label{app:friedman-dependence}
Each treatment is compared directly with the others inside a block, 
not with a pooled reference whose composition is fixed by an allocation, 
as in KW and logrank. Nor does the design leave an allocation free elsewhere: 
in a complete block design each of the $n$ blocks contributes exactly one observation 
per treatment, so every treatment is represented in a fraction $1/k$ of the $nk$ 
observations, whatever $n$ is. Both routes by which an allocation could reach 
the detection functional are therefore closed, and
$\theta(P) = (\mu_1,\ldots,\mu_k)$, with
$$
  \mu_j = 1 + \sum_{l \neq j} \bar{A}_{lj},
  \qquad
  \bar{A}_{lj} = \lim_{n\to\infty} \frac1n \sum_{i=1}^n
    \Bigl[ \mathrm{P}(X_{il} < X_{ij}) + \tfrac12\,\mathrm{P}(X_{il} = X_{ij}) \Bigr],
$$
the block-averaged within-block AUC of treatment~$l$ against treatment~$j$
(the index $l$ runs over treatments; $i$ remains the block index), 
satisfying $\bar{A}_{lj} + \bar{A}_{jl} = 1$. 
Each $\bar{A}_{lj}$ is a functional of the within-block joint distributions 
alone. $\Xi$ is a singleton.

\subsection{Logrank test}\label{App:logrank}

{Conventional calibration statistic:}
$Z = U_n/\sqrt{V_n}$, where
$U_n = \sum_j(d_{1j} - E_{1j})$,
(signed score form; the equivalent chi-squared form
 $\sum_g(O_g-E_g)^2/E_g \approx Z^2$ defines the same
 two-sided test),
asymptotically $N(0,1)$ under $H_0\mathpunct{:} S_1 = S_2$.

{Convenience version:} The logrank statistic is a
ratio statistic with implicit scaling. Rewriting to make
$a_n = \sqrt{n}$ explicit:
$$Z = \sqrt{n} \cdot \frac{U_n/n}{\sqrt{V_n/n}}$$
  
{Step 1 (Remove scaling):} Remove $a_n = \sqrt{n}$,
focusing on $(U_n/n)/\sqrt{V_n/n}$.

{Step 2 (Remove variance normalization):} Remove
$\sqrt{V_n/n}$, focusing on $U_n/n$.

{Step 3 (Remove nuisance):} In counting process
notation, $U_n/n$ can be written as:
$$
\frac{U_n}{n} = \frac{1}{n}\int \frac{Y_1(t)Y_2(t)}
{Y(t)}\left(\frac{dN_1(t)}{Y_1(t)} -
\frac{dN_2(t)}{Y_2(t)}\right)
$$
No nuisance components with $P$-free probability limits
remain, thus $U_n/n$ directly estimates
the detection functional.

{Step 4 (Minimal consistent estimator):}
The remaining quantity is $U_n/n$.

{Consistency check:} By the law of large numbers
for counting processes~\cite{andersen2012statistical}:
$$
\frac{U_n}{n} \xrightarrow{p} \Delta(\gamma) = \int w(t;\gamma)
[\lambda_1(t) - \lambda_2(t)]\,dt
$$
where
$w(t;\gamma) = \gamma(1-\gamma)y_1(t)y_2(t) /
\{\gamma y_1(t) + (1-\gamma) y_2(t)\}$,
$y_g(t) = \lim_{n\to\infty} Y_g(t)/n_g$ and $\gamma = \lim n_1/n$.
The weight already contains the at-risk functions $y_g$; in the
uncensored case $y_g(t) = S_g(t)$, so no further survival factor
enters $\Delta(\gamma)$.
Consistency passes. Note that $U_n$ itself diverges
at rate $n$; it is $U_n/n$ that converges to the
fixed population quantity $\Delta$.

{Detector:} $\hat{D}_n = U_n/n$, consistently
estimating the detection functional
$\theta(S_1,S_2;\gamma) = \Delta(\gamma) =
\int w(t;\gamma)[\lambda_1(t)-\lambda_2(t)]\,dt$.

{Minimality:} Immediate from
Lemma~\ref{Lem:minimality}; the attainable range is
$\Theta = \mathbb{R}$.

{Recovery of calibration statistic:}
$$Z = \frac{\sqrt{n}(U_n/n)}{\sqrt{V_n/n}} + o_p(1)$$
where $\sqrt{V_n/n}$ consistently estimates the
asymptotic standard deviation of $\sqrt{n}(U_n/n - \Delta)$
by martingale theory~\cite{andersen2012statistical}.

\subsubsection*{Auxiliary-argument dependence.}
\label{app:logrank-dependence}
Writing $y_i(t) = S_i(t)$ (no censoring), the martingale
representation of the logrank score gives the fixed-alternative
plim $w(t;\gamma) = \gamma(1-\gamma)y_1(t)y_2(t) /
\{\gamma y_1(t) + (1-\gamma)y_2(t)\}$. If $y_1(t) = k\,y_2(t)$ for
a $t$-constant $k$, $w(t;\gamma)$ factors as $c(\gamma)\,y_2(t)$
with $c(\gamma)>0$, and $\theta(P;\gamma)=0$ at one $\gamma$ iff
it is $0$ at every $\gamma$. 
This factoring condition is not an algebraic identity holding 
for arbitrary $(S_1,S_2)$: it requires the two survival functions 
to be proportional, which is non-generic (crossing survival curves violate it).
Consequently $\theta(P;\gamma)$ is generically $\gamma$-dependent.

\subsubsection*{Logrank detector under independent censoring}
\label{Note:logrank-censoring}
The logrank test is designed for right-censored
time-to-event data.
The exposition in Section~\ref{Sect:major} treats
the uncensored case for notational simplicity, but the
coherence verdict extends unchanged to the censored
setting under independent censoring.

Under independent censoring, the at-risk process
$Y_g(t) = \sum_{i} {1}(T_{gi} \geq t)$ is replaced
by the observed at-risk count, which under independent
censoring satisfies $Y_g(t)/n_g \xrightarrow{p} \pi_g(t)$,
the probability of being at risk at time~$t$ under the
joint distribution of event and censoring times.
The detector remains
$$
\hat{D}_n = \frac{O - E}{n}
= \int_0^\infty \frac{h_n(t)}{n}
\bigl\{d\hat{\Lambda}_1(t) - d\hat{\Lambda}_2(t)\bigr\},
$$
where $h_n(t) = Y_1(t)Y_2(t)/Y(t)$ is the at-risk weight,
and its probability limit under independent censoring is
$$
\theta(P;\gamma) =
\int_0^\infty \pi(t;\gamma)
\bigl\{\lambda_1(t) - \lambda_2(t)\bigr\} dt,
\qquad
\pi(t;\gamma) = \frac{\gamma(1-\gamma)\,\pi_1(t)\pi_2(t)}
{\gamma\,\pi_1(t)+(1-\gamma)\,\pi_2(t)},
$$
where $\lambda_g(t)$ are the hazard functions of the event
times and $\pi_g(t)$ the at-risk probabilities under the joint
event-censoring distribution. Censoring thus replaces the
at-risk functions $y_g$ of the uncensored case by $\pi_g$; the
form of the weight, and in particular its positivity and its
dependence on $\gamma$, is unchanged.

Thus, $\theta(P;\gamma) = 0$ holds whenever $\lambda_1 = \lambda_2$ $\pi$-almost
everywhere, but it holds as well at every configuration whose
$\pi(\cdot\,;\gamma)$-weighted hazard contrast cancels between
early and late times. Positivity of $\pi(t;\gamma)$ on the
observation window rules out sign-preserving cancellation only,
which is exactly why the ordered-hazards restriction of
Theorem~\ref{thm:logrank-ordered} restores coherence; it does
nothing against crossing hazards.

The detection-null set is accordingly
$$
\mathcal{D}_0(\gamma) = \Bigl\{P\mathpunct{:} \int_0^\infty w_\pi(t;\gamma)\{\lambda_1(t)-\lambda_2(t)\}\,dt = 0 \Bigr\} \;\supsetneq\; \bigl\{P\mathpunct{:} \lambda_1 = \lambda_2 \text{ a.e.\ on } \{t\mathpunct{:} w_\pi(t;\gamma) > 0\}\bigr\},
$$
where $w_\pi(t;\gamma)$ is the weight $w(t;\gamma)$ of Section~\ref{Subsect:logrank} 
with the at-risk functions $y_g$  replaced by $\pi_g$, the containment being strict 
for the same reason as in the uncensored case. The null $\mathcal{H}_0\mathpunct{:} S_1 = S_2$ 
under which the test is calibrated remains a strict subset of $\mathcal{D}_0(\gamma)$, 
and the blind spot is non-empty under censoring for the same crossing-hazard configurations. 
For the piecewise-exponential family of Lemma~\ref{lem:cstar} the proof carries over with $y_g(t;c)$ 
replaced by $\pi_g(t;c) = S_g(t;c)G_g(t)$: since $G_g \leq 1$ the domination bound $w_\pi(t;c,\gamma) \leq e^{-\min(a,b)t}$ is unchanged, 
and the two boundary limits retain their opposite signs whenever $G_1G_2 > 0$ on a set of positive Lebesgue measure, 
that is, whenever censoring leaves an observation window. Hence $c^*(\gamma)$ exists for every such censoring pair 
and every $\gamma \in (0,1)$, and the coherence verdict, detection incoherence of the logrank test, is unchanged by independent censoring.

%
\section*{Data availability statement}

R scripts reproducing the Monte Carlo studies in this paper are
available at \url{https://github.com/grendar/detection_coherence_mc}.

%
\section*{Acknowledgments}
The author thanks Edgar Brunner for drawing attention to
Brunner, Konietschke, Bathke and Pauly~\cite{brunner2021} in
correspondence following v1 of this preprint.

\section*{Declaration of Generative AI and AI-Assisted Technologies in the Writing Process}

During the preparation of tis work, the author used Claude Sonnet 4.6, 5 and Opus 5
(Anthropic, 2026), Kimi K2.6 (Moonshot AI, 2026)
and ChatGPT 5.6 Luna (OpenAI, 2026), to enhance readability of the text.
After using these tools the author reviewed and edited the content as 
necessary and takes full responsibility for the content and scientific integrity 
of the publication.

%
\bibliographystyle{amsplain}
\bibliography{references}

\end{document}